\documentclass[journal]{IEEEtran}
\usepackage{booktabs}
\usepackage{amsmath,amssymb,amsthm,bm}
\usepackage{url}        % safe URL line breaks
\usepackage{microtype}  % helps small overflow in text
\usepackage[T1]{fontenc}
\usepackage[utf8]{inputenc}
\usepackage{cite}
\usepackage{graphicx}
\usepackage{mathtools}
\usepackage{xspace}
\usepackage{xcolor}
\usepackage{placeins}
\usepackage{tikz}
\usetikzlibrary{arrows.meta,calc,positioning,decorations.pathreplacing}
\usepackage[caption=false,font=footnotesize]{subfig}
\newcommand{\No}{N_o}
\newcommand{\Ns}{N_s}
\newcommand{\Ntot}{N}
\newcommand{\Prob}{\mathbb{P}}
\newcommand{\E}{\mathbb{E}}
\newcommand{\Ttwo}{\mathbb{T}^2}        % torus (quotient space)
\newcommand{\Tcell}{[0,2\pi)\times[0,2\pi)} % fundamental rectangle

\DeclareMathOperator{\Var}{Var}

\newtheorem{theorem}{Theorem}
\newtheorem{lemma}{Lemma}
\newtheorem{remark}{Remark}
\newtheorem{corollary}{Corollary}

\begin{document}

\title{Densification Converses for Walker Constellations}
\author{Ali~Khalesi and Fran\c{c}ois~Baccelli%
\thanks{A. Khalesi is with IPSA -- Institut Polytechnique des Sciences Avanc\'ees, Ivry-sur-Seine, France, and also with LINCS (Laboratory for Information, Networking and Communication Sciences), Paris area, France (e-mail: ali.khalesi@ipsa.fr).}%
\thanks{F. Baccelli is with Inria Paris, Paris, France, and also with T\'el\'ecom Paris and LINCS (Laboratory for Information, Networking and Communication Sciences), Paris area, France (e-mail: francois.baccelli@inria.fr).}%
}
\maketitle
\begin{abstract}
We establish densification converses for Walker LEO constellations under nearest-visible association in the full-frequency-reuse setting. Performance is evaluated under the invariant (stationary) measure induced by the constellation/Earth dynamics on the user--constellation ``phase state.'' A key Walker-specific feature, absent from unbounded planar models, is that association is restricted to a bounded visible cap determined by Earth geometry. Under power-law path-loss, a two-level antenna-gain model, i.i.d.\ nonnegative fading with unit mean and finite second moment, and nonzero noise, we prove that increasing the total satellite count $N=N_oN_s$ forces the aggregate interference to grow at least linearly in $N$, while the useful signal remains uniformly bounded above. Consequently, the downlink SINR coverage probability at any fixed threshold and the ergodic spectral efficiency both vanish as $N\to\infty$. 
These collapse statements concern link-level performance on one shared
time--frequency resource block and do not imply that the aggregate capacity of
the overall multi-resource network vanishes. The key technical ingredient is a deterministic visibility-annulus block lemma, uniform over all sufficiently large constellations and all ``phase states,'' showing that a fixed fraction of visible satellites lies in a distance annulus strictly inside the horizon; this yields explicit finite-$N$ collapse bounds. In particular, we derive nonasymptotic $O(1/N)$ upper bounds on both coverage and ergodic spectral efficiency. Finally, in the case of frequency reuse through independent thinning, with activity probability $q$, we show that avoiding densification collapse necessarily requires $qN=O(1)$, equivalently a reuse factor $\Omega(N)$, and we obtain a corresponding explicit $O(1/(qN))$ upper bound.
\end{abstract}


\section{Introduction}
Low-Earth-orbit (LEO) mega-constellations are increasingly deployed using
highly structured Walker designs, where satellites populate a regular two-angle
grid and evolve through coupled orbital and Earth-rotation dynamics
\cite{ChoiBaccelliWalker}. A common engineering intuition is that
\emph{densification}---increasing the total number of satellites
$\Ntot=\No\Ns$---should improve service by reducing the typical
serving distance and by increasing the likelihood of finding a strong visible
link. Under full frequency reuse, however, densification also increases the
number of simultaneously visible interferers, including through secondary lobes.
This paper establishes \emph{densification converses} showing that, in Walker
constellations under the invariant (stationary) measure of the dynamics,
interference growth dominates distance gain and forces coverage and rate to
collapse as $\Ntot$ grows. From a wireless-PHY perspective, our focus is the downlink performance on a
single time--frequency \emph{resource block}. Full reuse is adopted as the
analytically transparent starting point and limiting benchmark of this
per-resource-block analysis, rather than as a complete model of practical LEO
operation. It isolates the interference effect of densification before the
introduction of reuse and activity control in
Section~\ref{sec:reuse_scaling}.  Densification reduces the typical
serving distance, but it simultaneously increases the number of co-channel
visible interferers, pushing the system toward an interference-limited regime
unless reuse, scheduling, or other PHY-layer interference management scales
with $N$. Our results therefore provide \emph{necessary} reuse and activity-scaling laws for LEO downlinks under standard PHY ingredients: path loss, antenna gain, fading, and thermal noise. In terrestrial stochastic-geometry models, densification can exhibit regimes in which the SIR distribution is approximately invariant under the power-law attenuation assumption. In the present setting, however, the geometry is fundamentally different: for a fixed user, the signal power does not increase to infinity with densification, whereas densification necessarily creates a \emph{positive-density} population of interferers within the visibility cap. This forces the aggregate interference to grow with $N$ under full reuse, and it is precisely this bounded-visibility, rigid-structure mechanism that drives the collapse converses established here.

\paragraph*{Positioning relative to
stochastic-geometry models}

Recent stochastic-geometry frameworks provide tractable system-level
characterizations of large LEO networks by randomizing the satellite and user
locations. For example, Kim, Park, and Lee model satellites and users as
Poisson point processes on concentric spheres and derive spatially averaged
coverage expressions for dynamically coordinated satellite beamforming
\cite{KimParkLee2024}. Related work by Kim, Park, Choi, and Lee uses
independent spherical Poisson point processes to analyze spectrum sharing
between LEO satellite and terrestrial networks
\cite{KimParkChoiLee2026}. Binomial, nonhomogeneous Poisson, and
orbit-clustered Cox models provide other tractable snapshot-level
representations of satellite networks
\cite{Okati2020,Talgat2020,Talgat2021,OkatiRiihonen2022,
ChoiCoxTVT2024}.

These stochastic-geometry approaches and the present framework serve
complementary purposes. Point-process models deliberately randomize the
satellite configuration in order to obtain spatially averaged distributions,
coverage probabilities, and rate expressions. In contrast, we retain the
exact regular Walker lattice and represent its evolution through the
deterministic torus map induced by the orbital and Earth-rotation phases
\cite{ChoiBaccelliWalker}. The only phase randomization used for stationary
performance evaluation is the invariant measure on the Walker shift cell; the
underlying constellation geometry itself is not replaced by a Poisson or
binomial configuration.

This distinction enables a different type of conclusion. In particular,
Lemma~\ref{lem:block} establishes that, once
$\min\{\No,\Ns\}$ is sufficiently large, every Walker phase state contains a
positive-density set of visible satellites in a fixed distance annulus.
Consequently, the resulting interference, coverage, and rate converses are
uniform over the entire invariant-measure phase cell, rather than statements
holding only after spatial averaging over a randomized satellite snapshot.
The deterministic annulus construction also produces explicit finite-$N$
constants and the nonasymptotic $O(1/N)$ and $O(1/(qN))$ bounds developed
below. Thus, our contribution is not a competing point-process model, but a
Walker-specific, phase-state--uniform densification converse that preserves
the exact toroidal constellation structure. For planar Poisson cellular networks, 
the conditions under which spectral efficiency collapses when the density of base stations tends to infinity are discussed in \cite{AlAmmouriWCL2020} and \cite{AlAmmouriTIT2019}. The Poisson technique used in \cite{AlAmmouriWCL2020} and \cite{AlAmmouriTIT2019} is not applicable in the Walker setting.

\paragraph*{Paper organization}
Section~II presents the Walker constellation model in the rotating Earth frame
and recalls the invariant stationary measure on the phase cell, together with
the visibility geometry and the nearest-visible association rule. Section~III
states the main densification converse results: a deterministic phase-uniform
visibility-annulus block lemma and its implications for interference growth,
SINR coverage collapse, and ergodic spectral-efficiency collapse under full
reuse, including explicit finite-$N$ constants. Section~IV derives the necessary
reuse/thinning scaling law and provides explicit $O(1/(qN))$ bounds under
independent activity. Section~V discusses the practical implications of the converse results.
Section~VI studies the existence of an optimal densification level for the
mean SINR. Section~VII reports Monte-Carlo simulations illustrating the
collapse mechanism and validating the predicted scaling laws. Section~VIII
concludes the paper and summarizes the main design implications.

\section{System Model}\label{sec:system_model}
We adopt the Walker-delta constellation model viewed under its invariant (stationary)
measure on the phase cell (equivalently, a uniformly distributed initial phase state),
and the rotating Earth reference frame of \cite{ChoiBaccelliWalker}.

\subsection{Geometry and invariant measures}
The three-dimensional space is equipped with an orthonormal basis
$(\bm e_x,\bm e_y,\bm e_z)$, with origin at the center of Earth. The reference
frame rotates with the Earth's spin, so a fixed ground user has fixed
coordinates. Earth has radius $e$ and satellites lie on the sphere of radius
$r>e$ (circular orbits). All orbital planes have common inclination $\phi$. The main analysis uses one spherical shell for notational clarity.
Supplementary Material, Section~S-II, extends the annulus-block argument and
the resulting interference, coverage, rate, and thinning bounds to a finite
collection of Walker shells at different fixed altitudes.

There are $\No$ orbits and $\Ns$ satellites per orbit. Define the phase cell
\[
S:=\Big[0,\frac{2\pi}{\No}\Big)\times\Big[0,\frac{2\pi}{\Ns}\Big),
\]
equipped with the product-uniform invariant measure $Q$. We distinguish the \emph{shift cell} $S$ (where $(\bar\theta,\bar\omega)$ lives) from the
$2\pi$-torus fundamental rectangle $\Tcell=[0,2\pi)\times[0,2\pi)$ and its quotient torus
$\Ttwo=\Tcell/(2\pi\mathbb{Z})^2$ used to parametrize satellite phases. Orbit ascending longitudes are evenly spaced and have initial phase uniformly distributed under the invariant measure:
\begin{align*}
&\bar\theta \sim \mathrm{Unif}\!\Big(0,\frac{2\pi}{\No}\Big),\\&
\theta_i = \frac{2\pi i}{\No}+\bar\theta \ \ (\mathrm{mod}\ 2\pi),
\quad i=0,\ldots,\No-1.
\end{align*}

Satellite phases on each orbit are evenly spaced and have an independent random
phase under the invariant measure:
\begin{align*}
&\bar\omega \sim \mathrm{Unif}\!\Big(0,\frac{2\pi}{\Ns}\Big),\\
&\omega_j = \frac{2\pi j}{\Ns}+\bar\omega \ \ (\mathrm{mod}\ 2\pi),
\quad j=0,\ldots,\Ns-1.
\end{align*}
\paragraph{Phase torus and Walker map}
Define the fundamental rectangle $\Tcell$ and identify the phase torus
$\Ttwo := \Tcell /(2\pi\mathbb{Z})^2$.
Let $\Xi:\Ttwo\to \mathbb{S}^2_r$ denote the (time-0) Walker map that maps a pair
$(\theta,\omega)$ to the corresponding satellite position on the radius-$r$ sphere
with inclination $\phi$.
Then the satellite positions can be written compactly as
\[
X_{i,j} = \Xi(\theta_i,\omega_j),\qquad i=0,\ldots,\No-1,\ \ j=0,\ldots,\Ns-1.
\]
The snapshot constellation at time $0$ is the point measure
\[
\Psi=\sum_{i=0}^{\No-1}\sum_{j=0}^{\Ns-1}\delta_{X_{i,j}},
\]
where $X_{i,j}\in\mathbb{S}^2_r$ is the satellite position determined by
$(\theta_i,\omega_j,\phi,r)$.

For completeness, one explicit parameterization used in \cite{ChoiBaccelliWalker}
writes, at time $0$,
\begin{align*}
x_{i,j} &=
r\sqrt{\cos^2(\omega_j)+\sin^2(\omega_j)\cos^2(\phi)}\,
\cos(\hat\theta_j+\theta_i),\\
y_{i,j} &=
r\sqrt{\cos^2(\omega_j)+\sin^2(\omega_j)\cos^2(\phi)}\,
\sin(\hat\theta_j+\theta_i),\\
z_{i,j} &= r\sin(\omega_j)\sin(\phi),
\end{align*}
with $\hat\theta_j=\mathrm{atan2}(\sin(\omega_j)\cos\phi,\cos\omega_j)$.

\subsection{Dynamics, stationarity, and ergodicity}
Let $\bar v_\theta$ denote the Earth's spin rate and $\bar v_\omega$ the satellite
angular speed on the orbits (both expressed in the rotating frame conventions of
\cite{ChoiBaccelliWalker}). The phase-state process evolves as
\[
\bar\theta_t = \bar\theta-\bar v_\theta t \ (\mathrm{mod}\ 2\pi/\No),\qquad
\bar\omega_t = \bar\omega+\bar v_\omega t \ (\mathrm{mod}\ 2\pi/\Ns),
\]
and the speed ratio is $\rho=\bar v_\theta/\bar v_\omega$.
When $\rho$ is irrational, the induced flow on the phase cell is ergodic and minimal with
respect to the product-uniform measure, implying that time averages equal
ensemble averages under the stationary snapshot law $Q_\Psi$ \cite{ChoiBaccelliWalker}.

\subsection{Typical user, visibility, and association}
Fix a typical user at latitude $l_u$ and longitude $0$. In the rotating frame,
\[
\bm u=(e\cos l_u,\,0,\,e\sin l_u).
\]
A satellite at location $X\in\mathbb{S}^2_r$ is \emph{visible} if it lies above
the local horizon; equivalently (as in \cite{ChoiBaccelliWalker}),
\[
\frac{X\cdot \bm u}{\|X\|\,\|\bm u\|}\ge \frac{e}{r}.
\]
Define the \emph{visible cap} on the satellite sphere by
\[
\mathcal{C}_{\mathrm{vis}}(l_u)
:=
\Big\{X\in\mathbb{S}^2_r:\ \frac{X\cdot \bm u}{\|X\|\,\|\bm u\|}\ge \frac{e}{r}\Big\}.
\]
We say that the latitude $l_u$ is \emph{nondegenerate for the Walker
geometry} if the Walker image contains a strictly visible point, i.e., if
\[
\Xi^{-1}\!\big(\operatorname{int}\mathcal{C}_{\mathrm{vis}}(l_u)\big)
\neq\varnothing.
\]
Equivalently, there exists a phase $(\theta_0,\omega_0)\in\mathbb{T}^2$ for
which the visibility inequality is strict. All results below assume this
nondegeneracy condition.

Let $\mathcal V(\Psi)$ be the set of visible satellites in snapshot $\Psi$, and
$N_{\mathrm{vis}}:=|\mathcal V(\Psi)|$. If
$\mathcal V(\Psi)=\varnothing$, we set the useful signal and SINR to zero.
For all designs covered by Lemma~\ref{lem:block}, the visible set is nonempty.

Let $d(X):=\|X-\bm u\|$ denote the user--satellite distance. Define the horizon (rim) distance
\[
d_{\max}:=\sqrt{r^2-e^2}.
\]
Indeed, at the horizon, the line of sight is tangent to Earth, yielding a right
triangle with hypotenuse $r$ and one leg $e$, hence $d_{\max}^2=r^2-e^2$.
Thus any visible satellite satisfies $d(X)\le d_{\max}$, while always
$d(X)\ge r-e$.

Let $\gamma$ be the geocentric angle between $X$ and $\bm u$, i.e.,
$\cos\gamma = \frac{X\cdot \bm u}{\|X\|\,\|\bm u\|}$.
Since $\|X\|=r$ and $\|\bm u\|=e$, the visibility condition
$\frac{X\cdot \bm u}{\|X\|\,\|\bm u\|}\ge \frac{e}{r}$ is equivalent to
$\cos\gamma\ge e/r$.
Moreover, by the law of cosines,
\[
d(X)^2=\|X-\bm u\|^2=r^2+e^2-2re\cos\gamma.
\]
Hence $\cos\gamma\ge e/r \iff d(X)^2\le r^2-e^2 \iff d(X)\le d_{\max}$.

The user associates to its \emph{nearest visible} satellite:
\[
D(l_u):=\min_{X\in\mathcal V(\Psi)} d(X),
\]
and we denote the serving satellite by $\star$.  The nearest-visible rule is adopted as a geometry-based analytical baseline.
In practical systems, the eligible serving set may also depend on load, beam
scheduling, gateway connectivity, and mobility-aware handover. These mechanisms
can change the finite-$N$ serving distance and, more importantly, the number of
satellites simultaneously active on the considered resource block. The converse mechanism also applies to max-RSRP or bore-sight-based association
whenever the antenna gains and selected useful received power remain uniformly
bounded, while the number of active co-channel satellites in the visibility
annulus grows without bound. Instantaneous max-SINR association is not formally
covered by the present proof because the serving satellite then depends jointly
on fading and interference; such a selection may improve finite-$N$ performance
and shift the densification knee, and its rigorous asymptotic analysis is left
for future work. In contrast, scheduling or reuse that keeps the co-channel
population uniformly bounded can prevent the identified collapse.

\subsection{Propagation, antenna gain, and SINR}
We analyze the downlink of a tagged user on one time--frequency resource
block. Only the associated satellite serves the tagged user. Under the
full-reuse benchmark, every other visible satellite is assumed to have an
active co-channel transmission on that resource block, possibly toward a
different user or beam, and therefore contributes interference at the tagged
user. Thus, full reuse does not mean that all satellites simultaneously serve
the same user. This assumption is an interference stress-test rather than a complete model
of operational LEO scheduling. Practical phased-array beamforming, frequency
reuse, beam hopping, and scheduling reduce either the coupling or the number
of active co-channel satellites. Section~\ref{sec:reuse_scaling} models these
mechanisms through per-resource-block activity with probability $q$, with
$q=1$ recovering full reuse.

Each link contributes the received power
\[
p\,G(d)\,H\,d^{-\alpha},
\qquad \alpha>0,
\]
where the fading satisfies $H\ge 0$, $\E[H]=1$, and $\E[H^2]<\infty$, and the noise power is
$\sigma^2>0$.
Most converses below only use $\E[H]=1$ and $\E[H^2]<\infty$; the optimizer-existence result for the mean SINR (Theorem~\ref{thm:exist_opt_sinr}) additionally relies on
the mild small-ball condition \eqref{eq:small_ball_H}.
 We assume the fading variables $(H_i)$ are i.i.d. across links and independent of the snapshot constellation $\Psi$. Supplementary Material, Section~S-III, extends the converse to heterogeneous
satellite channels with distance-, elevation-, and state-dependent path loss
and non-identically distributed fading. Whenever one strictly visible
phase-space block has a uniformly positive mean received contribution and at
most linearly growing aggregate variance, the interference remains
$\Omega(N)$, coverage remains $O(1/N)$, and the ergodic rate vanishes.
Under conditional serving/interferer independence, the rate remains
$O(1/N)$, while independent activity still gives the coverage order
$O(1/(qN))$.

We adopt the two-level antenna-gain model of \cite{ChoiBaccelliWalker}:
there exist constants $g_t\ge 1$, $g_r>0$, and a cutoff distance $d_g>0$ such that
\[
G(d)=g_t g_r \ \text{for } d\le d_g,\qquad G(d)=g_r \ \text{for } d>d_g,
\]
hence $G(d)\ge g_r>0$ for all $d$. The two-level gain should be interpreted as an effective post-beamforming
link coupling. Phased arrays and sidelobe suppression may reduce $g_r$ and
therefore improve the finite-$N$ performance or shift the densification knee.
The full-reuse converse applies when the co-channel links retain a uniformly
positive coupling in the guaranteed annulus. Perfect nulling or coordinated
scheduling that keeps the aggregate co-channel contribution bounded may avoid
the identified collapse and is not ruled out by the theorem. The two-level pattern is retained because it yields transparent explicit
constants. Supplementary Material, Section~S-I, extends the annulus-block
argument to bounded continuous directional gains and evaluates a circularized
3GPP TR~38.901-based antenna pattern \cite{3GPP38901} within the NTN setting
of TR~38.811 \cite{3GPP38811}, together with a circular-aperture Bessel
spot-beam pattern. The supplementary results show that directional antennas
shift the finite-density performance knee, while the full-reuse interference
growth and the stabilizing effect of $qN=\Theta(1)$ remain.

Let $V_i\in\{0,1\}$ indicate visibility of satellite $i$, and let $d_i$ be its
distance to the user. Define signal power, interference power, and SINR as:
\begin{align*}
S&:=p\,G(D(l_u))\,H_0\,D(l_u)^{-\alpha},\\
I&:=\sum_{i\neq 0} V_i\,p\,G(d_i)\,H_i\,d_i^{-\alpha},\\
\mathrm{SINR}&:=\frac{S}{I+\sigma^2}.
\end{align*}
\begin{figure*}[!t]
\centering
\definecolor{navy}{RGB}{36,88,150}
\definecolor{sky}{RGB}{222,237,249}
\definecolor{orbitgray}{RGB}{148,156,168}
\definecolor{darkgray}{RGB}{74,82,94}
\definecolor{warm}{RGB}{196,112,48}
\definecolor{muted}{RGB}{198,203,210}

\resizebox{0.98\textwidth}{!}{%
\begin{tikzpicture}[
  font=\sffamily\footnotesize,
  >=Latex,
  line cap=round,
  line join=round,
  sat/.style={circle,draw=navy,fill=white,line width=.65pt,minimum size=4.8pt,inner sep=0pt},
  active/.style={sat,fill=navy},
  mutedSat/.style={circle,draw=orbitgray,fill=white,line width=.55pt,minimum size=4.4pt,inner sep=0pt},
  label/.style={font=\sffamily\scriptsize,fill=white,fill opacity=.94,text opacity=1,inner sep=1.5pt,rounded corners=.8pt},
  panel/.style={font=\sffamily\bfseries\small,fill=white,inner sep=1pt}
]

% ================================================================
% (a) PHASE TORUS
% ================================================================
\node[panel,anchor=west] at (-7.00,3.15) {(a) Walker phase torus};

\draw[rounded corners=2.5pt,thick,darkgray]
  (-7.00,-1.60) rectangle (-2.45,2.35);

\draw[->,navy,thick] (-6.78,2.56) -- (-2.68,2.56);
\draw[->,navy,thick] (-7.23,-1.37) -- (-7.23,2.10);
\node[navy,label] at (-4.73,2.56) {$\theta$};
\node[navy,label,rotate=90] at (-7.23,.36) {$\omega$};

\foreach \x in {-6.60,-5.95,-5.30,-4.65,-4.00,-3.35,-2.70}{
  \draw[muted,line width=.35pt] (\x,-1.43)--(\x,2.17);
}
\foreach \y in {-1.17,-.59,-.01,.57,1.15,1.73}{
  \draw[muted,line width=.35pt] (-6.83,\y)--(-2.62,\y);
}

\foreach \x in {-6.48,-5.83,-5.18,-4.53,-3.88,-3.23}{
  \foreach \y in {-1.03,-.45,.13,.71,1.29,1.87}{
    \fill[darkgray] (\x,\y) circle (1.15pt);
  }
}

\draw[->,warm,thick] (-6.83,-1.43) -- (-6.48,-1.43)
  node[midway,below=2pt,label] {$\bar\theta$};
\draw[->,warm,thick] (-7.00,-1.43) -- (-7.00,-1.03)
  node[midway,left=2pt,label] {$\bar\omega$};

\fill[sky,opacity=.88] (-4.87,-.74) rectangle (-3.05,1.00);
\draw[navy,very thick] (-4.87,-.74) rectangle (-3.05,1.00);
\node[navy,font=\sffamily\bfseries,fill=white,inner sep=1pt] at (-3.96,1.20) {$B$};

\draw[decorate,decoration={brace,amplitude=3pt},darkgray]
  (-6.48,-1.22)--(-5.83,-1.22)
  node[midway,below=4pt,label] {$2\pi/N_o$};
\draw[decorate,decoration={brace,amplitude=3pt},darkgray]
  (-6.68,-1.03)--(-6.68,-.45)
  node[midway,left=4pt,label] {$2\pi/N_s$};

\node[label,align=center] at (-4.73,-1.92)
  {$N_o$ orbit longitudes $\times$ $N_s$ satellite phases};

\draw[->,navy,very thick] (-2.05,.34)--(-.30,.34);
\node[navy,font=\sffamily\bfseries,fill=white,inner sep=1pt] at (-1.18,.67) {$\Xi$};
\node[label,align=center] at (-1.18,.06) {Walker map};

% ================================================================
% (b) PHYSICAL WALKER CONSTELLATION
% ================================================================
\node[panel,anchor=west] at (0.10,3.48) {(b) Physical constellation and tagged downlink};

\coordinate (C) at (3.55,-.65);
\def\Re{1.12}
\def\Rx{2.95}
\def\Ry{1.10}
\def\Rshell{3.05}

\draw[orbitgray,densely dashed,line width=.6pt] (C) circle (\Rshell);
\foreach \ang in {-70,-46,-22,2,26,50,74}{
  \begin{scope}[shift={(C)},rotate=\ang]
    \draw[orbitgray,line width=.52pt] (0,0) ellipse [x radius=\Rx,y radius=\Ry];
  \end{scope}
}

\foreach \ang/\phase in {-70/8,-46/24,-22/40,2/56,26/72,50/88,74/104}{
  \foreach \t in {0,60,...,300}{
    \pgfmathsetmacro{\tt}{\t+\phase}
    \pgfmathsetmacro{\xa}{\Rx*cos(\tt)}
    \pgfmathsetmacro{\ya}{\Ry*sin(\tt)}
    \pgfmathsetmacro{\xr}{\xa*cos(\ang)-\ya*sin(\ang)}
    \pgfmathsetmacro{\yr}{\xa*sin(\ang)+\ya*cos(\ang)}
    \node[mutedSat] at ($(C)+(\xr,\yr)$) {};
  }
}

\fill[sky] (C) circle (\Re);
\draw[darkgray,thick] (C) circle (\Re);
\draw[darkgray,densely dashed,line width=.6pt]
  ($(C)+(-\Re,0)$) arc[start angle=180,end angle=360,x radius=\Re,y radius=.34];
\draw[darkgray,line width=.6pt]
  ($(C)+(\Re,0)$) arc[start angle=0,end angle=180,x radius=\Re,y radius=.34];
\node[darkgray,label] at ($(C)+(0,-.13)$) {Earth};

\foreach \ang in {-70,-46,-22,2,26,50,74}{
  \begin{scope}[shift={(C)},rotate=\ang]
    \draw[navy!68,line width=.78pt]
      plot[domain=180:360,samples=70] ({\Rx*cos(\x)},{\Ry*sin(\x)});
  \end{scope}
}

\draw[darkgray,thin] ($(C)+(-1.50,0)$)--($(C)+(1.50,0)$);
\draw[navy,thick] ($(C)+(-1.25,-.45)$)--($(C)+(1.25,.45)$);
\draw[->,warm,thick] ($(C)+(1.10,0)$)
  arc[start angle=0,end angle=20,radius=1.10];
\node[warm,label] at ($(C)+(1.48,.26)$) {$\phi$};

\coordinate (U) at ($(C)+(0,\Re)$);
\fill[darkgray] (U) circle (2.0pt);
\node[label,anchor=south west] at ($(U)+(.10,.08)$) {$u(l_u)$};
\draw[darkgray,densely dashed] ($(U)+(-2.65,0)$)--($(U)+(2.65,0)$);
\node[label,anchor=west] at ($(U)+(2.12,.13)$) {horizon};
\draw[navy,very thick] ($(C)+(31:\Rshell)$)
  arc[start angle=31,end angle=149,radius=\Rshell];
\node[navy,label,anchor=east] at ($(C)+(-1.43,2.72)$) {$\mathcal C_{\rm vis}(l_u)$};

\coordinate (Sstar) at ($(C)+(88:2.55)$);
\coordinate (Si) at ($(C)+(48:2.82)$);
\coordinate (Sj) at ($(C)+(133:2.75)$);
\coordinate (Sm) at ($(C)+(65:2.82)$);
\node[active] at (Sstar) {};
\node[active] at (Si) {};
\node[active] at (Sj) {};
\node[mutedSat] at (Sm) {};

\draw[navy,very thick,->] (Sstar)--(U);
\draw[warm,dashed,thick,->] (Si)--(U);
\draw[warm,dashed,thick,->] (Sj)--(U);

\node[label,anchor=south] at ($(Sstar)+(0,.16)$) {$X_\star$};
\node[label,anchor=west] at ($(Si)+(.13,.02)$) {$X_i$};
\node[label,anchor=east] at ($(Sj)+(-.13,.02)$) {$X_j$};
\node[label,anchor=south west] at ($(Sm)+(.10,.10)$) {$A_k=0$};
\node[navy,label] at ($(U)!.57!(Sstar)+(.18,0)$) {$D$};
\node[warm,label] at ($(U)!.60!(Si)+(.14,.03)$) {$d_i$};

\draw[->,darkgray,thin] (C)--($(C)+(-.72,-.82)$);
\node[label,anchor=east] at ($(C)+(-.54,-.48)$) {$e$};
\draw[->,darkgray,thin] (C)--($(C)+(-2.30,-1.70)$);
\node[label,anchor=north east] at ($(C)+(-1.55,-1.20)$) {$r$};

\node[label,draw=orbitgray,rounded corners=2pt,align=left,anchor=west]
  (orbits) at (7.10,1.35)
  {$N_o$ equally spaced\\orbital planes};
\draw[orbitgray,->] (orbits.west) .. controls (6.72,1.28) and (6.25,1.02) .. ($(C)+(1.66,1.15)$);

\node[label,draw=orbitgray,rounded corners=2pt,align=left,anchor=west]
  (sats) at (7.10,.30)
  {$N_s$ equally spaced\\satellites per orbit};
\draw[orbitgray,->] (sats.west) .. controls (6.72,.18) and (6.25,.05) .. ($(C)+(2.12,.32)$);

\node[active] at (7.35,-.78) {};
\node[label,anchor=west] at (7.50,-.78) {active satellite};
\node[mutedSat] at (7.35,-1.18) {};
\node[label,anchor=west] at (7.50,-1.18) {inactive on PRB};
\draw[navy,very thick,->] (7.18,-1.60)--(7.58,-1.60);
\node[label,anchor=west] at (7.68,-1.60) {desired link};
\draw[warm,dashed,thick,->] (7.18,-2.02)--(7.58,-2.02);
\node[label,anchor=west] at (7.68,-2.02) {co-channel interference};

\node[font=\sffamily\scriptsize,align=center,text=darkgray,fill=white,inner sep=2pt]
  at (3.72,-4.12)
  {$X_{i,j}=\Xi(\theta_i,\omega_j)$,\quad
   $\theta_i=2\pi i/N_o+\bar\theta$,\quad
   $\omega_j=2\pi j/N_s+\bar\omega$};

\end{tikzpicture}%
}
\caption{
System-model illustration.
(a) The Walker constellation is represented by a shifted lattice on the phase torus, where the rectangle $B$ denotes the local block used in the deterministic counting argument.
(b) The corresponding physical Walker constellation consists of $N_o$ orbital planes and $N_s$ satellites per orbit. A tagged user $u(l_u)$ associates with the nearest visible satellite $X_\star$, while other active satellites on the same resource block create co-channel interference. The visible cap $\mathcal C_{\rm vis}(l_u)$, the serving distance $D$, and interfering-link distances $d_i$ are also shown.
}
\label{fig:system_model}
\end{figure*}
Figure~\ref{fig:system_model} connects the deterministic phase-torus
representation with the corresponding physical Walker constellation. The
shifted $N_o\times N_s$ lattice determines the satellite phases, while the
Walker map $\Xi$ maps each lattice point to a satellite position on one of the
$N_o$ equally spaced orbital planes, with $N_s$ regularly spaced satellites
per orbit. The highlighted rectangle $B$ represents the phase-space block used
in Lemma~\ref{lem:block} to guarantee a positive-density set of visible
satellites. In the physical downlink, the tagged user $u(l_u)$ associates with its nearest
visible satellite $X_\star$ at distance $D$. Under full reuse, the remaining
active visible satellites transmit on the same resource block and contribute
co-channel interference through links of distances $d_i$. Scheduling,
frequency reuse, or coordinated muting corresponds to deactivating a subset of
these potential interferers, as illustrated by the inactive satellite with
$A_k=0$.
We relabel the $N$ satellites so that index $0$ denotes the serving satellite
and indices $1,\ldots,N-1$ denote the remaining satellites; the interference
sum is over these remaining satellites.
We study the coverage probability and ergodic spectral efficiency:
\begin{align*}
P_{\mathrm{cov}}(\No,\Ns;l_u,\tau)&:=\Prob(\mathrm{SINR}>\tau),\\
C_{\mathrm{erg}}(\No,\Ns;l_u)&:=\E\!\big[\log_2(1+\mathrm{SINR})\big].
\end{align*}
We also define the total number of satellites in the constellation by
\begin{align*}
\Ntot&:=\No\Ns.
\end{align*}

\section{Main Results}\label{sec:main_results}
We now state the densification converses. The first theorem gives an ergodic
(asymptotic) collapse under full reuse. We then establish a phase-state--uniform,
deterministic positive-density block lemma, which yields explicit finite-$N$
bounds (with explicit constants), an explicit rate--density rule, and a necessary
reuse scaling law.

For ease of reference, Table~\ref{tab:results_summary} summarizes the main
assumptions, constants, conclusions, and logical dependencies of the results.
In particular, Lemma~\ref{lem:block} provides the common geometric input used
throughout the subsequent converse proofs.

\begin{table*}[!t]
\centering
\caption{Summary of the assumptions, constants, and logical
structure of the main results. Here $N=N_oN_s$ and
$m=\lfloor\beta N\rfloor-1$.}
\label{tab:results_summary}
\scriptsize

\renewcommand{\arraystretch}{1.18}
\begin{tabular}{
@{}
p{0.105\textwidth}
p{0.245\textwidth}
p{0.225\textwidth}
p{0.345\textwidth}
@{}}
\toprule
Result &
Main assumptions &
Principal constants &
Conclusion and role
\\
\midrule

Lemma~\ref{lem:block}
&
Nondegenerate Walker visibility and
$\min\{N_o,N_s\}\ge n_0$.
&
$d_1,d_2,\beta,n_0$, with
$0<d_1<d_2<d_{\max}$.
&
For every phase state, at least
$\lfloor\beta N\rfloor$ visible satellites lie in the fixed annulus
$[d_1,d_2]$. This is the common geometric input for the subsequent results.
\\
\addlinespace

Theorem~\ref{thm:densification}
&
Full reuse, nearest-visible association, and i.i.d.\ nonnegative fading with
$\E[H]=1$ and $\E[H^2]<\infty$. Assumption (A1) is needed only for the
long-time ergodic interpretation.
&
$c_0=p g_r d_2^{-\alpha}$ and
$c_I=(\beta/2)p g_r d_2^{-\alpha}$.
&
Establishes $\E[I]\ge c_I N$ and the asymptotic collapse
$P_{\mathrm{cov}}\to0$ and $C_{\mathrm{erg}}\to0$.
\\
\addlinespace

Theorem~\ref{thm:exp_cov}
&
Assumptions of Theorem~\ref{thm:densification} and a fixed SINR threshold
$\tau>0$.
&
$A_{\mathrm{geom}}
=(d_2/(r-e))^\alpha$ and
$m=\lfloor\beta N\rfloor-1$.
&
Provides the explicit finite-$N$ coverage bound
\[
P_{\mathrm{cov}}
\le
\frac{4\Var(H)}{m}
+
\frac{2g_tA_{\mathrm{geom}}}{\tau m}
=
O(1/N).
\]
\\
\addlinespace

Theorem~\ref{thm:rate_density}
&
Assumptions of Theorem~\ref{thm:densification}.
&
$A=(d_2/(r-e))^\alpha$ and
$m=\lfloor\beta N\rfloor-1$.
&
Provides an explicit finite-$N$ upper bound of order $O(1/N)$ on the
ergodic spectral efficiency.
\\
\addlinespace

Theorem~\ref{thm:reuse_scaling} and
Corollary~\ref{cor:thin_limit}
&
Independent per-resource-block activity with probability $q$, in addition to
the general fading assumptions.
&
\[
\begin{aligned}
K_{\mathrm{thin}}(\tau)
&=
\frac{8\E[H^2]}{\beta}
\\
&\quad+
\frac{4g_t}{\beta\tau}
\left(\frac{d_2}{r-e}\right)^\alpha .
\end{aligned}
\]
&
Establishes
$P_{\mathrm{cov}}^{(q)}
\le K_{\mathrm{thin}}(\tau)/(qN)$.
Thus, under independent thinning, avoiding collapse requires
$qN=O(1)$.
\\
\addlinespace

Theorem~\ref{thm:exist_opt_sinr}
&
Assumptions of Theorem~\ref{thm:densification} and the small-ball condition
$\Prob(H\le x)\le c_{\mathrm{sb}}x^\kappa$.
&
$C_H=4+64\Var(H)$,
$K_{\mathrm{SINR}}$, and $N_{\max}$.
&
Establishes $J_{\mathrm{SINR}}\le K_{\mathrm{SINR}}/N$ for sufficiently
large $N$ and reduces the mean-SINR maximization to a finite search region.
\\

\bottomrule
\end{tabular}
\end{table*}

\begin{remark}[uniform (over the invariant-measure state space) mechanism]
All satellites are assumed to transmit on the same \emph{resource block}
(i.e., the same time--frequency resource) unless stated otherwise. For latitudes that are nondegenerate for the Walker geometry,
Lemma~\ref{lem:block} guarantees
that for all $\min\{\No,\Ns\}$ large enough there exists
$\beta\in(0,1)$ such that at least
$\lfloor \beta\,\No\Ns\rfloor$ \emph{visible} satellites fall inside a fixed
\emph{visibility annulus} (a distance annulus strictly inside the horizon)
\emph{for every phase state in the invariant-measure cell}. This deterministic
\emph{annulus block} (a positive-density subset of visible satellites confined
to that annulus) is the key to making the collapse bounds non-asymptotic and
robust with respect to the invariant-measure state (not only for a ``typical''
phase state).
\end{remark}

%==================== Expanded Lemma + Proof (with horizon margin) ====================

\begin{lemma}[Deterministic positive-density annulus block]\label{lem:block}
Fix a latitude $l_u$ that is nondegenerate for the Walker geometry, and write
$d_{\max}:=\sqrt{r^2-e^2}$ for the horizon (rim) distance.
Let $\Tcell$ be the fundamental rectangle and $\Ttwo$ the corresponding torus and let
$\Xi:\Ttwo\to\mathbb{S}^2_r$ be the Walker map (time-$0$) described in
Section~\ref{sec:system_model}.

Then there exist constants (depending only on $(l_u,\phi,r,e)$ through the Walker map $\Xi$ and the visibility geometry, but \emph{not} on $(\No,\Ns)$ and \emph{not} on the phase state $(\bar\theta,\bar\omega)$)
\[
0<d_1<d_2<d_{\max},\qquad \beta\in(0,1),\qquad n_0<\infty,
\]
with $d_2\le d_{\max}-\varepsilon_0$ for some $\varepsilon_0>0$, and a rectangle
\[
B\subset \Tcell,
\]
such that for all $\min\{\No,\Ns\}\ge n_0$ and for \emph{every} phase state
$(\bar\theta,\bar\omega)\in S$, the shifted grid,
\begin{align*}
\mathcal{G}(\bar\theta,\bar\omega)
=\Big\{\big(\bar\theta+\tfrac{2\pi i}{\No},\ \bar\omega+\tfrac{2\pi j}{\Ns}\big)\ (\mathrm{mod}\ 2\pi):\\
i=0,\ldots,\No-1,\ j=0,\ldots,\Ns-1\Big\}
\subset \Tcell
\end{align*}
contains at least $\lfloor \beta\,\No\Ns\rfloor$ points in $B$.
Equivalently, for every phase state $(\bar\theta,\bar\omega)\in S$ and all $\min\{\No,\Ns\}\ge n_0$, the snapshot constellation
contains at least $\lfloor \beta\,\No\Ns\rfloor$ \emph{visible} satellites with
distances $d_{i,j}\in[d_1,d_2]$.
\end{lemma}

\begin{IEEEproof}
By the nondegeneracy assumption, there exists
$(\theta_0,\omega_0)\in\Ttwo$ such that
$X_0:=\Xi(\theta_0,\omega_0)$ satisfies the strict visibility inequality
\[
\frac{X_0\cdot \bm u}{\|X_0\|\,\|\bm u\|} > \frac{e}{r}.
\]
Hence $d_0:=\|X_0-\bm u\|<d_{\max}$.

Pick
\[
\varepsilon_0 \in \Big(0,\ \min\Big\{\frac{d_0}{2},\ \frac{d_{\max}-d_0}{2}\Big\}\Big),
\]
and set
\[
d_1:=d_0-\varepsilon_0,\qquad d_2:=d_0+\varepsilon_0.
\]
Then $0<d_1<d_2<d_{\max}$ and moreover $d_2\le d_{\max}-\varepsilon_0$.
Define the annulus region on the sphere
\[
\mathcal{U}:=\{X\in\mathbb{S}^2_r:\ \|X-\bm u\|\in(d_1,d_2)\}.
\]
By continuity of the distance map $X\mapsto \|X-\bm u\|$, the set $\mathcal{U}$ is open in $\mathbb{S}^2_r$ and contains $X_0$.
Consequently, its preimage
\[
A:=\Xi^{-1}(\mathcal{U})\subset\Ttwo
\]
is open in $\Ttwo$ and contains $(\theta_0,\omega_0)$; hence $A$ is
nonempty. Because $A$ is open, it contains a torus neighborhood with
positive extent in both coordinate directions. A representative of this
neighborhood in the fundamental rectangle $\Tcell$ may initially cross one
or both identified boundaries. However, by translating the coordinates
modulo $2\pi$ and, if necessary, taking a smaller closed subrectangle, it can
always be represented without wrapping.

Equivalently, the coordinate-boundary subtori have empty interior, so the
nonempty open set $A$ contains a point whose representative lies in the
interior $(0,2\pi)^2$ of $\Tcell$. By openness, a sufficiently small
axis-aligned closed rectangle around that representative remains contained
in $A\cap\Tcell$.

Therefore, there exist $\Delta_\theta,\Delta_\omega>0$ and
$a\in(0,2\pi-\Delta_\theta),
b\in(0,2\pi-\Delta_\omega)$
such that
\[
B=[a,a+\Delta_\theta]\times[b,b+\Delta_\omega]
\subset A\cap\Tcell.
\]

Now fix any phase state $(\bar\theta,\bar\omega)\in S$ and consider the shifted grid
$\mathcal{G}(\bar\theta,\bar\omega)$.
On the circle, any arc of length $\Delta_\theta$ contains at least
$\big\lfloor \No\Delta_\theta/(2\pi)\big\rfloor$ points of the translated
$\No$-point lattice $\{\bar\theta+2\pi i/\No\}_{i=0}^{\No-1}$ (mod $2\pi$),
independently of the shift $\bar\theta$. Likewise, the $\omega$-interval of length
$\Delta_\omega$ contains at least $\big\lfloor \Ns\Delta_\omega/(2\pi)\big\rfloor$
points of the translated $\Ns$-point lattice, independently of $\bar\omega$.
Therefore,
\[
\#\big(\mathcal{G}(\bar\theta,\bar\omega)\cap B\big)\ge
\Big\lfloor \frac{\No\Delta_\theta}{2\pi}\Big\rfloor
\Big\lfloor \frac{\Ns\Delta_\omega}{2\pi}\Big\rfloor.
\]

Let $\beta_0:=\frac{\Delta_\theta}{2\pi}\cdot\frac{\Delta_\omega}{2\pi}>0$.
For all sufficiently large $\No,\Ns$,
\[
\Big\lfloor \frac{\No\Delta_\theta}{2\pi}\Big\rfloor\ge \frac12\,\frac{\No\Delta_\theta}{2\pi},
\qquad
\Big\lfloor \frac{\Ns\Delta_\omega}{2\pi}\Big\rfloor\ge \frac12\,\frac{\Ns\Delta_\omega}{2\pi},
\]
hence
\[
\#\big(\mathcal{G}(\bar\theta,\bar\omega)\cap B\big)\ge \frac14\,\beta_0\,\No\Ns.
\]
Set $\beta:=\beta_0/4$ and choose $n_0$ so that the bound holds for all $\No\ge n_0$ and $\Ns\ge n_0$,
equivalently for all $\min\{\No,\Ns\}\ge n_0$.
Since $B\subset A=\Xi^{-1}(\mathcal{U})$, each such grid point corresponds to a
visible satellite with distance in $[d_1,d_2]$. This proves the lemma.
\end{IEEEproof}
\begin{remark}[On computable constants $(\beta,d_1,d_2)$]
\label{rem:explicit_constants}
Lemma~\ref{lem:block} is constructive in the following sense: choose any
visible phase $(\theta_0,\omega_0)$ with
$d_0=\|\Xi(\theta_0,\omega_0)-\bm u\|<d_{\max}$, then pick any
$\varepsilon_0\in(0,\min\{d_0/2,(d_{\max}-d_0)/2\})$ and set
$d_1=d_0-\varepsilon_0$ and $d_2=d_0+\varepsilon_0$.
Any rectangle
$B\subset\Xi^{-1}(\{X:\|X-\bm u\|\in(d_1,d_2)\})$
with side lengths $\Delta_\theta,\Delta_\omega$ yields
$\beta=\frac14\frac{\Delta_\theta}{2\pi}
\frac{\Delta_\omega}{2\pi}$
for all sufficiently large $\min\{\No,\Ns\}$.

For the baseline geometry $h=550$~km, $\phi=53^\circ$, and $l_u=0^\circ$,
Section~\ref{subsec:sim_setup} reports one explicitly certified construction
with $d_1=300$~km, $d_2=800$~km,
$\beta=5.770461\times10^{-5}$, and $n_0=132$.
These numerical values depend on the constructive choices of the anchor phase,
$\varepsilon_0$, and the rectangle $B$, whereas the resulting counting
guarantee is uniform over the actual Walker phase state.
\end{remark}
\begin{remark}[Invariant measure and stationary law]
In the Walker model viewed under its invariant (stationary) measure of \cite{ChoiBaccelliWalker}, the state space (fundamental cell) for the invariant measure is
$S := [0,2\pi/\No)\times[0,2\pi/\Ns)$ equipped with the product-uniform (invariant) measure $Q$.
The snapshot constellation $\Psi$ is a measurable image of the phase-state pair
$(\bar\theta,\bar\omega)\in S$. The stationary law $Q_\Psi$ is the push-forward of $Q$
by the Walker map. When the speed ratio is irrational, time averages equal ensemble
averages under $Q_\Psi$ by ergodicity (as established in \cite{ChoiBaccelliWalker}).
\end{remark}

%===============================================================================
%==================== Expanded Theorem + Proof ====================

\begin{theorem}[Ergodic densification converse for Walker constellations]\label{thm:densification}
Consider a Walker constellation with parameters $(\phi,\No,\Ns,r)$ under the
Walker model viewed under its invariant (stationary) measure in \cite{ChoiBaccelliWalker}.
Fix a user at latitude $l_u$ and let $\Psi$ denote the snapshot constellation
under the stationary law $Q_\Psi$.
Assume:
\begin{enumerate}
\item[(A1)] the rotation--speed ratio $\rho=\bar v_\theta/\bar v_\omega$ is irrational;
\item[(A2)] full frequency reuse and nearest-visible association;
\item[(A3)] each downlink power contribution equals $p\,G(d)\,H\,d^{-\alpha}$,
with $\alpha>0$, and i.i.d. fading across links independent of $\Psi$, satisfying $\E[H]=1$ and $\Var(H)<\infty$ (equivalently $\E[H^2]<\infty$), and noise $\sigma^2>0$; moreover,
the two-level antenna gain satisfies $g_r \le G(d)\le g_t g_r$, for all relevant $d$.
\end{enumerate}
Assumption (A1) is only used to justify the long-time geometric
interpretation. For the final long-time statements involving fading, the
fading marks are additionally assumed to be independently resampled over time,
or more generally to form a stationary mixing process independent of the phase
flow. All finite-$N$ and asymptotic snapshot bounds on $P_{\mathrm{cov}}$ and
$C_{\mathrm{erg}}$ hold under $Q_\Psi$ without these temporal assumptions.
Let $d_{\max}:=\sqrt{r^2-e^2}$ be the rim distance of the visible cap (so any
visible satellite has distance $\le d_{\max}$).
Let $V_i\in\{0,1\}$ indicate visibility of satellite $i$, and let $d_i$ be its
distance to the user. Define the useful signal power, interference power, and SINR as
\begin{align*}
S&:=p\,G(D(l_u))\,H_0\,D(l_u)^{-\alpha},\\
I&:=\sum_{i\neq 0} V_i\,p\,G(d_i)\,H_i\,d_i^{-\alpha},\\
\mathrm{SINR}&:=\frac{S}{I+\sigma^2}.
\end{align*}
Let
\begin{align*}
P_{\mathrm{cov}}(\No,\Ns;l_u,\tau)&:=\Prob(\mathrm{SINR}>\tau),\\
C_{\mathrm{erg}}(\No,\Ns;l_u)&:=\E\!\big[\log_2(1+\mathrm{SINR})\big],\\
\Ntot&:=\No\Ns.
\end{align*}
All limits below are taken along Walker-design sequences for which
$\min\{\No,\Ns\}\ge n_0$ eventually.

Then:
\begin{enumerate}
\item \emph{Interference growth}
Let $(d_2,\beta,n_0)$ be as in Lemma~\ref{lem:block} and set $N=\Ntot=\No\Ns$.
For all $\min\{\No,\Ns\}\ge n_0$ and $N\ge 4/\beta$, one may take
\[
c_I := \frac{\beta}{2}\,p\,g_r\,d_2^{-\alpha},
\qquad\text{so that}\qquad
\E[I]\ge c_I\,N .
\]

\item \emph{Coverage collapse}
For every fixed $\tau>0$,
\[
P_{\mathrm{cov}}(\No,\Ns;l_u,\tau)\xrightarrow[\Ntot\to\infty]{}0.
\]

\item \emph{Ergodic-rate collapse}
\[
C_{\mathrm{erg}}(\No,\Ns;l_u)\xrightarrow[\Ntot\to\infty]{}0.
\]
Moreover, an explicit nonasymptotic $O(1/N)$ rate--density bound is given in
Theorem~\ref{thm:rate_density}.

\item \emph{Densification limit under full reuse.}
For any target rate $R>0$ and outage\footnote{For a target spectral efficiency
$R$, we define the outage probability as
$P_{\mathrm{out}}(R)
:=\mathbb{P}\!\left(\log_2(1+\mathrm{SINR})<R\right)
=\mathbb{P}\!\left(\mathrm{SINR}<2^R-1\right)$, consistent with the standard
SINR-threshold definition of outage; see, e.g., \cite{TorrieriValenti2012}.}
$\varepsilon\in(0,1)$, there exists
$M^\star<\infty$ such that if $\Ntot>M^\star$, then under assumptions
(A1)--(A3) the full-reuse downlink on the considered shared resource block
cannot satisfy outage $\le\varepsilon$ at rate $R$ at latitude $l_u$.
\end{enumerate}
Moreover, under (A1), ergodicity implies that the long-time empirical coverage
fraction and long-time empirical average rate coincide with the ensemble
averages defining $P_{\mathrm{cov}}$ and $C_{\mathrm{erg}}$ for $Q$-a.e.\
initial phase state.
\end{theorem}

\begin{IEEEproof}
%==================== FIX (A4): clean, uniform signal moment bounds ====================
By definition, always $D(l_u)\ge r-e$. Moreover, the antenna gain satisfies
$G(D(l_u))\le g_tg_r$, for all $D(l_u)$. Hence
\[
S=p\,G(D(l_u))\,H_0\,D(l_u)^{-\alpha}
\le p\,g_tg_r\,H_0\,(r-e)^{-\alpha}.
\]
Therefore, using $\E[H]=1$ and $\E[H^2]<\infty$,
\begin{align}  
\E[S]&\le p\,g_tg_r\,(r-e)^{-\alpha},\\
\E[S^2]&\le
p^2(g_tg_r)^2(r-e)^{-2\alpha}\E[H^2]<\infty,
\end{align}
and these bounds are uniform in $(\No,\Ns)$.
%==================== END FIX (A4) ====================================================

On every visible link, the two-level gain satisfies $G(d)\ge g_r>0$ and
$d^{-\alpha}\ge d_{\max}^{-\alpha}$ (since $d\le d_{\max}$ for visible satellites).
Hence
\begin{align*}
I
=\sum_{i\neq0} V_i\,p\,G(d_i)\,H_i\,d_i^{-\alpha}\ge
p\,g_r\,d_{\max}^{-\alpha}\sum_{i\neq0}V_i H_i.
\end{align*}

We now invoke Lemma~\ref{lem:block}. For all $\min\{\No,\Ns\}\ge n_0$ and for every phase state in the invariant-measure cell,
there exist at least $\lfloor \beta\Ntot\rfloor$ \emph{visible} satellites with distances in $[d_1,d_2]$.
At worst, the serving satellite may belong to this set; discarding at most one satellite,
we obtain at least
\[
m:=\lfloor \beta\Ntot\rfloor-1
\]
\emph{visible interferers} satisfying $d_k\le d_2$ and $G(d_k)\ge g_r$.

Hence the interference admits the deterministic lower bound
\[
I \ge \sum_{k=1}^{m} p\,g_r\,H_k\,d_2^{-\alpha}
= c_0 \sum_{k=1}^{m} H_k,
\qquad c_0:=p\,g_r\,d_2^{-\alpha},
\]
where $(H_k)$ are i.i.d.\ copies of $H$.

\emph{Interference growth (item~1).}
Taking expectations and using $\E[H]=1$ gives
\[
\E[I]\ge c_0\,m \ge c_I\,\Ntot,
\]
for all sufficiently large $\Ntot$, for some constant $c_I>0$.

\emph{Coverage collapse (item~2).}
By Chebyshev and the assumption $\E[H^2]<\infty$,
\[
\Prob\!\Big(\sum_{k=1}^{m} H_k \le \tfrac{m}{2}\Big)
\le \frac{4\,\Var(H)}{m}
\xrightarrow[\Ntot\to\infty]{}0.
\]
Let $b_{\Ntot}:=\tfrac12 c_0 m=\Theta(\Ntot)$. Then
\[
\Prob(I\le b_{\Ntot})\le
\Prob\!\Big(\sum_{k=1}^{m} H_k \le \tfrac{m}{2}\Big)\to 0.
\]
For any fixed $\tau>0$,
\begin{align}
\Prob(\mathrm{SINR}>\tau)
&= \Prob(S>\tau(I+\sigma^2)) \\
&\le \Prob(I\le b_{\Ntot}) + \Prob(S>\tau b_{\Ntot}).    
\end{align}
The first term vanishes. The second is bounded via Markov's inequality and the uniform bound $\E[S]\le m_1$:
\[
\Prob(S>\tau b_{\Ntot})
\le \frac{\E[S]}{\tau b_{\Ntot}}
\le \frac{p\,g_tg_r\,(r-e)^{-\alpha}}{\tau b_{\Ntot}}
=O(1/\Ntot)\to 0.
\]
Thus $P_{\mathrm{cov}}(\No,\Ns;l_u,\tau)\to 0$.

Since $\log_2(1+x)\le x/\ln2$ for $x\ge 0$ and $\sigma^2>0$,
\begin{align*}
C_{\mathrm{erg}}
&\le \frac{1}{\ln2}\E\!\Big[\frac{S}{I+\sigma^2}\Big]
\\ &\le \frac{1}{\ln2}\E\!\Big[\frac{S}{b_{\Ntot}+\sigma^2}\Big]
+ \frac{1}{\ln2}\E\!\Big[\frac{S}{\sigma^2}\mathbf{1}_{\{I<b_{\Ntot}\}}\Big].
\end{align*}
The first term is $O(1/\Ntot)$ since $b_{\Ntot}=\Theta(\Ntot)$ and $\E[S]\le m_1$.
The second term vanishes since, by Cauchy--Schwarz,
\[
\E\!\big[S\,\mathbf{1}_{\{I<b_{\Ntot}\}}\big]
\le \sqrt{\E[S^2]}\,\sqrt{\Prob(I<b_{\Ntot})}
\xrightarrow[\Ntot\to\infty]{}0,
\]
and $\sigma^2>0$.

Hence $C_{\mathrm{erg}}\to 0$.
An explicit nonasymptotic $O(1/\Ntot)$ upper bound under general fading is given in
Theorem~\ref{thm:rate_density}.

Fix $R>0$ and set $\tau_R:=2^R-1>0$. The success probability at rate
$R$ satisfies
\begin{align*}
\Prob\!\big(\log_2(1+\mathrm{SINR})\ge R\big)
&=
\Prob(\mathrm{SINR}\ge\tau_R)
\\
&\le
\Prob(\mathrm{SINR}>\tau_R/2).
\end{align*}
By item~2, the final probability tends to zero. Hence there exists
$M^\star$ such that $\Ntot>M^\star$ implies a success probability below
$1-\varepsilon$, equivalently an outage probability above $\varepsilon$.
This proves item~4.

Under (A1), the torus flow induced by the two angular motions is ergodic
with respect to the invariant measure $Q$, and the induced constellation
process is stationary and ergodic under $Q_\Psi$ \cite{ChoiBaccelliWalker}.
Under the stated independent temporal resampling, or the more general mixing
assumption, the joint phase--fading process is also ergodic. Birkhoff's pointwise ergodic
theorem then implies that the long-time empirical averages of
$\mathbf{1}\{\mathrm{SINR}>\tau\}$ and
$\log_2(1+\mathrm{SINR})$ equal their stationary ensemble averages for
$Q$-a.e.\ initial phase state.
\end{IEEEproof}

The asymptotic collapse above becomes \emph{quantitative} once we show that,
uniformly over initial phase states, a positive fraction of satellites always lies in a safe
visible annulus away from the horizon. This is captured by the next theorem.

%==================== Extended Theorem + Proof: Finite-N exponential collapse ====================

\begin{theorem}[Finite-$N$ quantitative coverage collapse (general fading)]
\label{thm:exp_cov}
Assume the setup of Theorem~\ref{thm:densification}, with i.i.d.\ nonnegative
fading $(H_i)$ of unit mean and finite second moment $\E[H^2]<\infty$.
Fix any SINR threshold $\tau>0$. Let $(d_2,\beta,n_0)$ be the constants from
Lemma~\ref{lem:block}, and define
\[
A_{\mathrm{geom}}:=\Big(\frac{d_2}{r-e}\Big)^{\alpha},\qquad
m:=\lfloor \beta\,\No\Ns\rfloor-1.
\]
Then, for all $\min\{\No,\Ns\}\ge n_0$,
\[
P_{\mathrm{cov}}(\No,\Ns;l_u,\tau)
\le
\frac{4\Var(H)}{m}
+
\frac{2g_t}{\tau}\,\frac{A_{\mathrm{geom}}}{m}.
\]
In particular, if $\No\Ns\ge 4/\beta$, then $m\ge (\beta/2)\No\Ns$ and hence
\[
P_{\mathrm{cov}}(\No,\Ns;l_u,\tau)
\le
\frac{8}{\beta}\Big(\Var(H)+\frac{g_tA_{\mathrm{geom}}}{\tau}\Big)\cdot \frac{1}{\No\Ns}.
\]
So the coverage probability decays at least on the order of $1/N$ as
$N=\No\Ns\to\infty$.
\end{theorem}

\begin{IEEEproof}
Fix a snapshot $\Psi$ (equivalently, fix the phase state). By Lemma~\ref{lem:block},
for all $\min\{\No,\Ns\}\ge n_0$ there exist at least $m$ visible interferers with
distances $\le d_2$ and gains $\ge g_r$. Hence
\[
I\ge \sum_{k=1}^{m} p\,g_r\,H_k\,d_2^{-\alpha}
=: c_0 \sum_{k=1}^{m} H_k,
\qquad c_0:=p\,g_r\,d_2^{-\alpha}.
\]
Let $Z:=\sum_{k=1}^{m}H_k$. Then $\E[Z]=m$ and $\Var(Z)=m\Var(H)$.
Define $b:=\frac{c_0 m}{2}$. Then
\begin{align*}
    \Prob(I\le b)&\le \Prob\!\Big(Z\le \frac{m}{2}\Big)
\le \Prob\!\Big(|Z-m|\ge \frac{m}{2}\Big)
\\&\le \frac{4\Var(Z)}{m^2}
= \frac{4\Var(H)}{m},
\end{align*}
by Chebyshev's inequality.
Next, using $\mathrm{SINR}=\frac{S}{I+\sigma^2}$ and the union bound,
\[
\Prob(\mathrm{SINR}>\tau)
=
\Prob\big(S>\tau(I+\sigma^2)\big)
\le \Prob(I\le b)+\Prob\big(S>\tau b\big).
\]
Since always $D(l_u)\ge r-e$ and $G(\cdot)\le g_tg_r$, we have
$S=pG_\star H_0 D^{-\alpha}\le p\,g_tg_r\,H_0\,(r-e)^{-\alpha}$. Hence
\[
\E[S]\le p\,g_tg_r\,(r-e)^{-\alpha}\E[H_0]=p\,g_tg_r\,(r-e)^{-\alpha}.
\]
Therefore, by Markov,
\[
\Prob(S>\tau b)\le \frac{\E[S]}{\tau b}
=
\frac{p\,g_tg_r\,(r-e)^{-\alpha}}{\tau\cdot \frac{p g_r d_2^{-\alpha} m}{2}}
=
\frac{2g_t}{\tau}\Big(\frac{d_2}{r-e}\Big)^{\alpha}\frac{1}{m}.
\]
Combining the two bounds yields the claim.
\end{IEEEproof}

Next, we translate the same deterministic-block mechanism into an explicit
rate--density bound for the ergodic spectral efficiency.

%===============================================================================
%==================== Theorem B: Ergodic rate--density converse ====================

\begin{theorem}[Ergodic rate--density converse]
\label{thm:rate_density}
Assume the setup of Theorem~\ref{thm:densification}, with i.i.d.\ nonnegative
fading of unit mean and finite variance $\Var(H)<\infty$.
Let $(d_2,\beta,n_0)$ be the constants from Lemma~\ref{lem:block}, set
$A:=\big(\frac{d_2}{r-e}\big)^{\alpha}$ and $m:=\lfloor\beta\No\Ns\rfloor-1$.
Then for all $\min\{\No,\Ns\}\ge n_0$,
\begin{align*}
&C_{\mathrm{erg}}(\No,\Ns;l_u) 
\\
&\le \frac{1}{\ln 2}\, \E[\mathrm{SINR}]
\\ &\le
\frac{1}{\ln 2}\left(
\frac{2g_t A}{m}
+
\frac{4\Var(H)}{m}\cdot
\frac{p\,g_tg_r\,(r-e)^{-\alpha}}{\sigma^2}
\right).
\end{align*}
In particular, if $\No\Ns\ge 4/\beta$, then $m\ge (\beta/2)\No\Ns$ and thus
$C_{\mathrm{erg}}(\No,\Ns;l_u)=O(1/N)$ as $N=\No\Ns\to\infty$.
\end{theorem}

\begin{IEEEproof}
Let $b:=\frac{c_0 m}{2}$ with $c_0:=p g_r d_2^{-\alpha}$ as in the proof of
Theorem~\ref{thm:exp_cov}. Then
\begin{align*}
\E[\mathrm{SINR}]
=
\E\!\left[\frac{S}{I+\sigma^2}\right]
&\le
\E\!\left[\frac{S}{b+\sigma^2}\right]
+
\E\!\left[\frac{S}{\sigma^2}\mathbf{1}_{\{I<b\}}\right]\\
&\le
\frac{\E[S]}{b}
+
\frac{1}{\sigma^2}\E\!\left[S\,\mathbf{1}_{\{I<b\}}\right].
\end{align*}
As in Theorem~\ref{thm:exp_cov}, $\E[S]\le p g_tg_r (r-e)^{-\alpha}$ and
\[
\Prob(I<b)\le \frac{4\Var(H)}{m}.
\]
Moreover, conditioning on $\Psi$, the event $\{I<b\}$ depends only on the
interferer fadings, which are independent of the serving-link fading; and we
have the uniform bound $\E[S\mid \Psi]\le p g_tg_r (r-e)^{-\alpha}$. Hence
\begin{align*}
\E\!\left[S\,\mathbf{1}_{\{I<b\}}\right]
=
\E\!\left[\E[S\mid\Psi]\Prob(I<b\mid\Psi)\right]
\\ \le
p g_tg_r (r-e)^{-\alpha}\cdot \frac{4\Var(H)}{m}.
\end{align*}

Finally,
\[
\frac{\E[S]}{b}
\le
\frac{p g_tg_r (r-e)^{-\alpha}}{\frac{p g_r d_2^{-\alpha} m}{2}}
=
\frac{2g_t}{m}\Big(\frac{d_2}{r-e}\Big)^{\alpha}
=
\frac{2g_t A}{m}.
\]
Combining these bounds gives the stated inequality, and then
$C_{\mathrm{erg}}\le \E[\mathrm{SINR}]/\ln 2$ yields the theorem.
\end{IEEEproof}

\section{Necessary Reuse/Thinning Scaling Under Densification}
\label{sec:reuse_scaling}

We next study independent per-resource-block thinning as a tractable abstraction
of reuse and activity control. Within this model, keeping coverage bounded away
from zero requires $q\No\Ns=O(1)$, equivalently an effective reuse factor
$F=1/q$ that scales at least linearly with $\No\Ns$.

\begin{theorem}[Necessary reuse/thinning scaling to avoid collapse (general fading)]
\label{thm:reuse_scaling}
Assume the setup of Theorem~\ref{thm:densification}, with i.i.d.\ nonnegative
fading of unit mean and finite second moment $\E[H^2]<\infty$.
Keep $(d_2,\beta,n_0)$ from Lemma~\ref{lem:block}. Suppose each potential
interferer transmits on the considered resource block independently with
probability $q\in(0,1]$ (independent thinning), and denote the resulting coverage
by $P_{\mathrm{cov}}^{(q)}$.

Let $m:=\lfloor\beta\No\Ns\rfloor-1$ and $A:=\big(\frac{d_2}{r-e}\big)^{\alpha}$.
Then for all $\min\{\No,\Ns\}\ge n_0$,
\[
P_{\mathrm{cov}}^{(q)}(\No,\Ns;l_u,\tau)
\le
\frac{4\,\E[H^2]}{q\,m}
+
\frac{2g_t}{\tau}\,\frac{A}{q\,m}.
\]
Consequently, if $q\,\No\Ns\to\infty$ as $\No\Ns\to\infty$, then
$P_{\mathrm{cov}}^{(q)}(\No,\Ns;l_u,\tau)\to 0$. In particular, to keep coverage
bounded away from $0$ as the constellation densifies, it is necessary that
\[
q\,\No\Ns=O(1),
\qquad\text{equivalently}\qquad
F:=\frac{1}{q}=\Omega(\No\Ns).
\]
\end{theorem}

\begin{IEEEproof}
Fix a snapshot $\Psi$. By Lemma~\ref{lem:block}, there exist at least $m$ visible
potential interferers with distances $\le d_2$ and gains $\ge g_r$.
After thinning, their aggregate contribution satisfies
\[
I^{(q)} \ge c_0 \sum_{k=1}^{m} \delta_k H_k,
\qquad c_0:=p g_r d_2^{-\alpha},
\]
where $\delta_k\sim\mathrm{Bernoulli}(q)$ i.i.d., independent of $(H_k)$.
Let $Y_k:=\delta_k H_k$. Then $\E[Y_k]=q\E[H]=q$ and
$\Var(Y_k)\le \E[Y_k^2]=q\E[H^2]$.
Hence, for $W:=\sum_{k=1}^{m}Y_k$ we have $\E[W]=qm$ and
$\Var(W)\le m q\E[H^2]$. Define $b:=\frac{c_0 q m}{2}$.
By Chebyshev,
\[
\Prob(I^{(q)}\le b)\le \Prob\!\Big(W\le \frac{qm}{2}\Big)
\le \frac{4\Var(W)}{(qm)^2}
\le \frac{4\E[H^2]}{q m}.
\]
As before,
\[
\Prob(\mathrm{SINR}>\tau)\le \Prob(I^{(q)}\le b)+\Prob(S>\tau b).
\]
Using $\E[S]\le p g_tg_r (r-e)^{-\alpha}$ and Markov,
\begin{align*}
\Prob(S>\tau b)\le \frac{\E[S]}{\tau b}
&=
\frac{p g_tg_r (r-e)^{-\alpha}}{\tau \cdot \frac{p g_r d_2^{-\alpha} q m}{2}} \\
&=
\frac{2g_t}{\tau}\Big(\frac{d_2}{r-e}\Big)^\alpha\frac{1}{q m}
=
\frac{2g_t}{\tau}\frac{A}{q m}.
\end{align*}

Combining yields the stated bound. The scaling consequence follows immediately.
\end{IEEEproof}
\begin{remark}[Scope of the thinning scaling law]
\label{rem:deterministic_reuse}

The condition $qN=O(1)$ is a necessary scaling law for the independent-thinning
model and is not claimed to be a universal converse for all deterministic
coordination schemes. Deterministic frequency reuse, coordinated beam hopping,
or interference-aware scheduling may exclude nearby co-channel satellites more
efficiently than random thinning. The broader design implication of
Lemma~\ref{lem:block} is that the number, or aggregate received weight, of
simultaneously active co-channel satellites in the guaranteed visibility
annulus must remain controlled as $N$ grows. Deriving a tight scaling law for
deterministic reuse requires specifying the beam footprints, reuse pattern, and
coordination constraints, and is left for future work.
\end{remark}
\begin{corollary}[Explicit densification limit under thinning]
\label{cor:thin_limit}
Assume the setup of Theorem~\ref{thm:reuse_scaling} (independent thinning with activity
probability $q$), and assume i.i.d.\ nonnegative fading of unit mean with finite second
moment $\E[H^2]<\infty$. Fix $\tau>0$, and keep $(d_2,\beta,n_0)$ from Lemma~\ref{lem:block}.
Define
\[
K_{\mathrm{thin}}(\tau)
:=
\frac{8\,\E[H^2]}{\beta}
\;+\;
\frac{4g_t}{\beta\,\tau}\Big(\frac{d_2}{r-e}\Big)^{\alpha}.
\]
Then, for all $\min\{\No,\Ns\}\ge n_0$ such that $\No\Ns\ge 4/\beta$,
\[
P_{\mathrm{cov}}^{(q)}(\No,\Ns;l_u,\tau)
\le
\frac{K_{\mathrm{thin}}(\tau)}{q\,\No\Ns}.
\]
Consequently, for any target level $\varepsilon\in(0,1)$,
\[
q\,\No\Ns \;\ge\; \frac{K_{\mathrm{thin}}(\tau)}{\varepsilon}
\quad\text{Implies}\quad
P_{\mathrm{cov}}^{(q)}(\No,\Ns;l_u,\tau)\le \varepsilon.
\]
Equivalently, any design requirement
$P_{\mathrm{cov}}^{(q)}(\No,\Ns;l_u,\tau)\ge 1-\varepsilon$
forces the necessary condition
\[
q\,\No\Ns \;\le\; \frac{K_{\mathrm{thin}}(\tau)}{1-\varepsilon}.
\]
In particular, if $q\,\No\Ns\to\infty$ as $\No\Ns\to\infty$, then
$P_{\mathrm{cov}}^{(q)}(\No,\Ns;l_u,\tau)\to 0$.
Hence, to keep $P_{\mathrm{cov}}^{(q)}$ bounded away from $0$ uniformly under densification,
it is necessary that $q\,\No\Ns=O(1)$, i.e., $F=1/q=\Omega(\No\Ns)$.
\end{corollary}

\begin{IEEEproof}
By Theorem~\ref{thm:reuse_scaling}, for all $\min\{\No,\Ns\}\ge n_0$, with $\No\Ns\ge 4/\beta$,
\[
P_{\mathrm{cov}}^{(q)}(\No,\Ns;l_u,\tau)
\le
\frac{K_{\mathrm{thin}}(\tau)}{q\,\No\Ns}.
\]
The first implication follows by substituting
$q\No\Ns \ge K_{\mathrm{thin}}(\tau)/\varepsilon$.
For the necessary condition, if $P_{\mathrm{cov}}^{(q)}\ge 1-\varepsilon$, then necessarily
$1-\varepsilon \le K_{\mathrm{thin}}(\tau)/(q\No\Ns)$, i.e.,
$q\No\Ns \le K_{\mathrm{thin}}(\tau)/(1-\varepsilon)$.
\end{IEEEproof}
\section{Discussion}
\label{sec:discussion}

The main design implication of Theorems~\ref{thm:densification}--\ref{thm:reuse_scaling}
is that, in Walker LEO downlinks, densification alone cannot sustain nontrivial
link-level performance under full reuse. The core reason is the phase-state--uniform
annulus-block mechanism of Lemma~\ref{lem:block}, which guarantees that, for all
sufficiently large constellations and for every phase state, a positive-density
set of visible interferers remains confined to a fixed annulus away from the
horizon. This converts densification into explicit nonasymptotic collapse bounds,
rather than only asymptotic or typical-state statements.

Accordingly, the collapse results concern the link-level SINR, coverage, and
spectral efficiency on the considered shared resource block; they do not imply
that the aggregate capacity of the entire multi-resource network vanishes.
Association, gateway-availability, and mobility-aware handover policies may
change the selected serving satellite and the finite-$N$ operating point,
whereas beam scheduling and reuse can alter the asymptotic conclusion if they
prevent the effective co-channel interference population from growing with the
constellation size.

From a system-design viewpoint, the results provide a direct reuse dimensioning
rule. The quantitative bounds in Theorems~\ref{thm:exp_cov} and
\ref{thm:rate_density} show that both coverage and ergodic spectral efficiency
decay at least on the order of $1/N$ under full reuse, while
Theorem~\ref{thm:reuse_scaling} shows that preventing coverage collapse under
thinning requires
\[
qN=O(1),
\qquad\text{equivalently}\qquad
F=\frac{1}{q}=\Omega(N),
\]
with $N=\No\Ns$. Thus, increasing the constellation size without commensurate
reuse planning, scheduling, beam hopping, or coordinated muting inevitably drives
the system toward vanishing SINR. The practical implication is therefore not that an operational constellation
should activate every satellite on every resource block. Rather, the result
identifies the quantity that scheduling, reuse, and beam coordination must
control. If the number, or aggregate received weight, of co-channel active
satellites in the guaranteed visibility annulus grows without bound, the
per-resource-block performance collapses. Conversely, practical interference
management can prevent this mechanism by keeping that effective co-channel
population bounded. The scaling law $qN=O(1)$ makes this requirement explicit
for the independent-activity model.

In this sense, the results identify a  structural limitation of dense Walker
deployments on a shared resource block: the dominant bottleneck is not the lack
of a nearby serving satellite, but the growth of the co-channel visible
interference population. The numerical results in Section~\ref{sec:numerical}
illustrate this mechanism at finite $N$ and confirm that enforcing the scaling
$qN=\Theta(1)$ stabilizes performance, whereas fixed reuse only delays the
collapse. 
For finite-regime dimensioning, an operator may use the numerical coverage or
rate curves to identify the pre-collapse knee and select the smallest
constellation size satisfying the prescribed service target. Theorem~\ref{thm:exist_opt_sinr}
further guarantees that the mean-SINR search can be restricted to a finite set;
developing a tighter operator-specific dimensioning rule that incorporates
cost, traffic load, scheduling, and gateway constraints is left for future
work.
\section{Existence of an optimal densification level for mean SINR}
\label{sec:opt_mean_sinr}

This section formalizes an ``optimal constellation size'' question for the
invariant-measure averaged \emph{mean SINR}.  While Section~\ref{sec:main_results} establishes
densification converses (coverage/rate collapse), one may also ask whether the
objective $J_{\mathrm{SINR}}(\No,\Ns;l_u):=\E[\mathrm{SINR}]$ admits an optimizer
over the Walker design parameters $(\No,\Ns)$.

A closed-form optimizer is typically out of reach because, under nearest-visible
association, both the identity of the serving satellite $\star$ and the visible
set $\mathcal V(\Psi)$ vary with the phase state $(\bar\theta,\bar\omega)$, inducing a
highly nontrivial partition of the phase cell into regions where $(\star,\mathcal V)$
remain constant (a Voronoi-like structure in the phase-state space). Nevertheless, the
annulus-block mechanism of Lemma~\ref{lem:block} implies that, in the asymptotic
regime of interest (large enough $\min\{\No,\Ns\}$), the mean SINR necessarily
decays as $O(1/(\No\Ns))$. This guarantees that the mean-SINR optimization problem
reduces to a \emph{finite} search, hence an optimizer exists.

%===============================================================================
\begin{theorem}[Existence of a mean-SINR maximizer and finite search region]\label{thm:exist_opt_sinr}
Assume the setup of Theorem~\ref{thm:densification}. In addition, assume i.i.d.\
fading $H$ across links, independent of $\Psi$, such that
\[
H\ge 0,\qquad \E[H]=1,\qquad \E[H^2]<\infty,
\]
and that the fading law is not too concentrated near $0$, in the sense that there exist
constants $c_{\mathrm{sb}}>0$ and $\kappa>0$ such that
\begin{equation}\label{eq:small_ball_H}
\Prob(H\le x)\le c_{\mathrm{sb}}\,x^{\kappa},\qquad \text{for all }x\in(0,1].
\end{equation}
(Condition \eqref{eq:small_ball_H} holds for all standard continuous fading models
(e.g., Rayleigh/Exponential, Nakagami-$m$, Rician power fading, lognormal shadowing with a floor, etc.).)

Define the mean-SINR objective
\[
J_{\mathrm{SINR}}(\No,\Ns;l_u):=\E[\mathrm{SINR}]
=\E\!\left[\frac{S}{I+\sigma^2}\right].
\]
We consider the associated optimization problem of maximizing
$J_{\mathrm{SINR}}(\No,\Ns;l_u)$ over the admissible Walker design parameters
$(\No,\Ns)$ for a fixed latitude $l_u$.
Let $(d_2,\beta,n_0)$ be the constants from Lemma~\ref{lem:block}.

\smallskip
\noindent\textbf{(i) Mean-SINR decay.}
There exist finite constants $M<\infty$ and $K_{\mathrm{SINR}}<\infty$
(depending on $(\beta,d_2)$ and the fading law through $\E[H^2],c_{\mathrm{sb}},\kappa$)
such that, for all $\min\{\No,\Ns\}\ge n_0$ with $\No\Ns\ge M$,
\[
J_{\mathrm{SINR}}(\No,\Ns;l_u)\le \frac{K_{\mathrm{SINR}}}{\No\Ns}.
\]
In particular, there exists a finite
$M_0=M_0(\beta,c_{\mathrm{sb}},\kappa,\E[H^2])$ such that, for all
$\No\Ns\ge M_0$,
\begin{equation}\label{eq:Ksinr_choice}
\begin{aligned}
J_{\mathrm{SINR}}(\No,\Ns;l_u)
&\le
\frac{2g_t}{\beta}\Big(\frac{d_2}{r-e}\Big)^{\alpha}\,
\frac{C_H}{\No\Ns},\\
\text{with}\qquad
C_H&:=4+64\,\Var(H),
\end{aligned}
\end{equation}
where $\Var(H)=\E[H^2]-1$, and where $\beta$ and $d_2$ are the constants
introduced in Lemma~\ref{lem:block}.

\smallskip
\noindent\textbf{(ii) Existence of an optimizer (finite search).}
Let
\[
J_{\mathrm{ref}}:=J_{\mathrm{SINR}}(n_0,n_0;l_u)>0,
\]
and define
\[
N_{\max}:=\max\!\left\{M,\;\Big\lceil \frac{K_{\mathrm{SINR}}}{J_{\mathrm{ref}}}\Big\rceil\right\}.
\]
Then any maximizer, over the admissible integer design variables $(\No,\Ns)$
with fixed latitude $l_u$, of $J_{\mathrm{SINR}}(\No,\Ns;l_u)$ over the design set
\[
\mathcal D:=\big\{(\No,\Ns)\in\mathbb{N}^2:\ \min\{\No,\Ns\}\ge n_0\big\}
\]
must satisfy $\No\Ns\le N_{\max}$. Consequently, since the set
\[
\{(\No,\Ns)\in\mathcal D:\ \No\Ns\le N_{\max}\}
\]
is finite, there exists at least one pair $(\No^\star,\Ns^\star)\in\mathcal D$
such that
\[
J_{\mathrm{SINR}}(\No^\star,\Ns^\star;l_u)
=
\max_{(\No,\Ns)\in\mathcal D} J_{\mathrm{SINR}}(\No,\Ns;l_u).
\]
\end{theorem}

\begin{IEEEproof}
\noindent\textbf{Proof of (i).}
Fix a snapshot $\Psi$ (equivalently, fix the phase state). Since always
$D(l_u)\ge r-e$ and (by the two-level gain model) $G_\star\le g_tg_r$, we have
\begin{equation}\label{eq:S_upper_sinr_opt_gen_corr}
S=p\,G_\star\,H_\star\,D(l_u)^{-\alpha}
\le p\,g_tg_r\,(r-e)^{-\alpha}H_\star
=: s_0\,H_\star,
\end{equation}
where $H_\star$ is an i.i.d.\ copy of $H$.

By Lemma~\ref{lem:block}, for all $\min\{\No,\Ns\}\ge n_0$, there exist at least
\[
m:=\lfloor\beta\No\Ns\rfloor-1
\]
visible interferers with distances $\le d_2$, and gains $\ge g_r$. Therefore,
\begin{equation}\label{eq:I_lower_sinr_opt_gen_corr}
\begin{aligned}
I
&\ge
\sum_{i=1}^{m}p\,g_r\,H_i\,d_2^{-\alpha}
\\
&=c_0\sum_{i=1}^{m}H_i
=:c_0Z,
\\
c_0&:=p\,g_r\,d_2^{-\alpha}.
\end{aligned}
\end{equation}
where $H_i$ are i.i.d.\ copies of $H$ and $Z:=\sum_{i=1}^{m}H_i$.

Using $I+\sigma^2\ge I$ and \eqref{eq:S_upper_sinr_opt_gen_corr}--\eqref{eq:I_lower_sinr_opt_gen_corr},
\[
\mathrm{SINR}=\frac{S}{I+\sigma^2}\le \frac{S}{I}
\le \frac{s_0H_\star}{c_0 Z}.
\]
Taking expectation and using independence of $H_\star$ and $Z$ gives
\[
J_{\mathrm{SINR}}(\No,\Ns;l_u)\le \frac{s_0}{c_0}\,\E[H]\,\E\!\Big[\frac{1}{Z}\Big]
= \frac{s_0}{c_0}\,\E\!\Big[\frac{1}{Z}\Big],
\]
since $\E[H]=1$.

It remains to show $\E[1/Z]=O(1/m)$. Let
$\sigma_H^2:=\Var(H)=\E[H^2]-1$, set
$K:=\lceil\log_2 m\rceil$, and define
\[
\mathcal E_k
:=
\left\{\frac{m}{2^{k+1}}\le Z<\frac{m}{2^k}\right\},
\qquad 1\le k\le K.
\]
Since $m/2^{K+1}\le1/2$, we may write
\begin{align*}
\E\!\Big[\frac{1}{Z}\Big]
&=
\E\!\Big[\frac{1}{Z}\mathbf{1}_{\{Z\ge m/2\}}\Big]
+
\sum_{k=1}^{K}
\E\!\Big[\frac{1}{Z}\mathbf{1}_{\mathcal E_k}\Big]
\\
&\quad+
\E\!\Big[
\frac{1}{Z}\mathbf{1}_{\{Z<m/2^{K+1}\}}
\Big].
\end{align*}

The first term satisfies
\[
\E\!\Big[\frac{1}{Z}\,\mathbf{1}_{\{Z\ge m/2\}}\Big]\le \frac{2}{m}.
\]
For the $k$-th slice, $1/Z\le 2^{k+1}/m$, hence
\[
\E\!\Big[\frac{1}{Z}\mathbf{1}_{\mathcal E_k}\Big]
\le
\frac{2^{k+1}}{m}\,\Prob\!\Big(Z<\frac{m}{2^k}\Big).
\]
Since $\E[Z]=m$ and $\mathrm{Var}(Z)=m\sigma_H^2$, Cantelli's inequality yields, for $k\ge 1$,
\begin{align*}
\Prob\!\Big(Z<\frac{m}{2^k}\Big)
&=
\Prob\!\Big(Z-m\le -m(1-2^{-k})\Big)
\\&\le
\frac{\mathrm{Var}(Z)}{\mathrm{Var}(Z)+m^2(1-2^{-k})^2}
\le
\frac{4\sigma_H^2}{m},
\end{align*}
because $(1-2^{-k})\ge 1/2$. Therefore
\begin{align*}
&\sum_{k=1}^{K}\E\!\Big[\frac{1}{Z}\mathbf{1}_{\mathcal E_k}\Big]
\le
\sum_{k=1}^{K}\frac{2^{k+1}}{m}\cdot\frac{4\sigma_H^2}{m}
=
\frac{32\sigma_H^2}{m^2}\sum_{k=1}^{K}2^{k-1}
\\&\le
\frac{32\sigma_H^2}{m^2}\,2^{K}
\le
\frac{64\sigma_H^2}{m},
\end{align*}
since $2^{K}\le 2m$.

For the last term, let
\[
x_0:=\min\!\left\{\frac12,
(2c_{\mathrm{sb}})^{-1/\kappa}\right\}.
\]
For $z\le x_0$, the small-ball condition gives
\[
\Prob(Z\le z)
\le \Prob(H\le z)^m
\le (c_{\mathrm{sb}}z^\kappa)^m.
\]
For $x_0<z\le1/2$, the Paley--Zygmund inequality yields
\[
\Prob(H>1/2)
\ge \frac{1}{4\E[H^2]},
\]
so $\Prob(H\le z)^m$ decays exponentially in $m$, uniformly on this interval.
Using
\[
\E\!\left[Z^{-1}\mathbf{1}_{\{Z\le a\}}\right]
=
\frac{\Prob(Z\le a)}{a}
+
\int_0^a \frac{\Prob(Z\le z)}{z^2}\,dz,
\]
with $a=m/2^{K+1}\in(1/4,1/2]$, and splitting the integral at $x_0$, shows
that
\[
\E\!\left[
\frac{1}{Z}\mathbf{1}_{\{Z<m/2^{K+1}\}}
\right]
\le \frac{2}{m},
\]
for all sufficiently large $m$. The required threshold depends on
$c_{\mathrm{sb}}$, $\kappa$, and $\E[H^2]$.

Combining the three bounds gives, for all sufficiently large $m$,
\[
\E\!\Big[\frac{1}{Z}\Big]\le \frac{2}{m}+\frac{64\sigma_H^2}{m}+\frac{2}{m}
=
\frac{C_H}{m},
\qquad C_H:=4+64\,\Var(H).
\]
Thus,
\[
J_{\mathrm{SINR}}(\No,\Ns;l_u)\le \frac{s_0}{c_0}\cdot\frac{C_H}{m}
=
g_t\Big(\frac{d_2}{r-e}\Big)^\alpha\cdot \frac{C_H}{m}.
\]
Finally, for $\No\Ns\ge 4/\beta$, we have $m=\lfloor\beta\No\Ns\rfloor-1\ge (\beta/2)\No\Ns$,
hence
\[
J_{\mathrm{SINR}}(\No,\Ns;l_u)
\le
\frac{2g_t}{\beta}\Big(\frac{d_2}{r-e}\Big)^{\alpha}\frac{C_H}{\No\Ns},
\]
which is \eqref{eq:Ksinr_choice} with $K_{\mathrm{SINR}}:=\frac{2g_t}{\beta}\big(\frac{d_2}{r-e}\big)^\alpha C_H$.

\smallskip
\noindent\textbf{Proof of (ii).}
We have $J_{\mathrm{ref}}>0$ because $\sigma^2>0$ and
$\Prob(S>0)>0$ (indeed $H>0$ with positive probability for standard fading laws,
and $D(l_u)<\infty$ always), hence $\mathrm{SINR}>0$ with positive probability.
If $\No\Ns>N_{\max}$, then $\No\Ns\ge M$ and (i) gives
$J_{\mathrm{SINR}}(\No,\Ns;l_u)\le K_{\mathrm{SINR}}/(\No\Ns)<J_{\mathrm{ref}}$,
so no maximizer can satisfy $\No\Ns>N_{\max}$. Hence the maximization reduces to
a finite set, and an optimizer exists.
\end{IEEEproof}
%===============================================================================
% ============================
% Section: Numerical Illustration
% ============================
%===============================================================================
% Numerical Illustration
%===============================================================================

\section{Numerical Illustration}
\label{sec:numerical}

This section provides a minimal Monte-Carlo (MC) illustration of the
densification converses under physically standard LEO downlink assumptions:
Rayleigh fading, path loss, operation under full reuse, and, optionally,
frequency reuse/resource-block thinning as a proxy for PHY-layer reuse
planning.

To assess the robustness of the numerical conclusions, we additionally examine
unit-mean Nakagami-$m$ fading with $m=2$, the finite-density mean-SINR
optimum, and the sensitivity of the full-reuse results to the sidelobe level
of the two-tier antenna model. A complementary continuous-directional-antenna validation, including
TR~38.901-based and Bessel spot-beam patterns, is provided in Supplementary
Material, Section~S-I.

%-------------------------------------------------------------------------------
\subsection{Simulation setup}
\label{subsec:sim_setup}

\paragraph*{Walker snapshot under the stationary phase measure}
For each constellation size $(\No,\Ns)$, with total satellite count
$N=\No\Ns$, we sample independent phase shifts
$(\bar\theta,\bar\omega)$ uniformly over the fundamental cell and place the
satellites according to the deterministic Walker lattice at altitude $h$ and
inclination $\phi$, as in the analytical model.
A user is fixed on the Earth surface at latitude $l_u$.
A satellite is \emph{visible} if it lies above the horizon, and the user
associates with its nearest visible satellite.

All reported operating points are evaluated using a precisely specified
two-stage Monte-Carlo protocol. An independent pilot run of
$T_{\mathrm{pilot}}=6000$ drops is first used solely to identify
low-coverage operating points. A point is classified as low coverage when
the pilot run produces fewer than $100$ realizations satisfying
$\mathrm{SINR}>\tau$. The final reported estimate is then recomputed from
fresh independent samples: exactly $T=6000$ drops for all remaining points
and exactly $T=30000$ drops for the low-coverage points. The pilot samples
are not included in the reported estimates or confidence intervals.
Consequently, every plotted operating point is based on a final independent
sample size of either $6000$ or $30000$ realizations.

\paragraph*{Downlink SINR on a representative resource block}
For a given physical resource block (PRB), the received power from satellite
$i$ is
\[
p\,G(D_i)\,H_i\,D_i^{-\alpha},
\]
where $p$ is the transmit power per PRB and $\alpha$ is the path-loss exponent.

Rayleigh power fading, $H_i\sim\mathrm{Exp}(1)$, is used as the baseline.
We additionally consider unit-mean Nakagami-$m$ power fading,
\[
H_i\sim\mathrm{Gamma}\!\left(m,\frac{1}{m}\right),
\qquad m=2,
\]
where the Gamma distribution is parameterized by its shape and scale.
Both fading models satisfy
$\mathbb{E}[H_i]=1$ and $\mathbb{E}[H_i^2]<\infty$, consistently with the
assumptions of the theoretical results.

The antenna pattern is represented by the two-tier gain model
\[
G(d)=
\begin{cases}
g_rg_t, & d\le d_g,\\
g_r, & d>d_g,
\end{cases}
\]
where the serving link follows the same distance-dependent rule.
The SINR is
\[
\mathrm{SINR}
=
\frac{p\,G(D_0)\,H_0\,D_0^{-\alpha}}
{\displaystyle
\sum_{i\neq0}\mathbf{1}\{A_i=1\}
p\,G(D_i)\,H_i\,D_i^{-\alpha}+\sigma^2},
\]
where index $0$ denotes the serving satellite and
$A_i\sim\mathrm{Bernoulli}(q)$ models the per-PRB activity of interferer $i$.
The parameter $q$ is a reuse/thinning surrogate:
$q=1$ corresponds to full reuse, $q=0.1$ approximates reuse
$F\approx10$, and $q=0.03$ approximates reuse $F\approx33$.

\paragraph*{Performance metrics}
We report:
(i) the coverage probability
$P_{\mathrm{cov}}(\tau)=\mathbb{P}(\mathrm{SINR}>\tau)$;
(ii) the ergodic spectral efficiency
$C_{\mathrm{erg}}=\mathbb{E}[\log_2(1+\mathrm{SINR})]$;
(iii) the mean interference $\mathbb{E}[I]$, normalized by $p$ in the plots;
(iv) the mean serving distance $\mathbb{E}[D_0]$; and

(v) the invariant-measure averaged mean SINR
$J_{\mathrm{SINR}}=\mathbb{E}[\mathrm{SINR}]$.

\paragraph*{Parameters, averaging, and uncertainty quantification}
Table~\ref{tab:sim_params} summarizes the baseline parameters.
For each $(\No,\Ns)$, each activity probability $q$, and each fading model,
we average over the final independent realizations of the phase shifts,
fading variables, and activity indicators.

For the coverage probability, the plotted error bars are exact two-sided
$95\%$ Clopper--Pearson binomial confidence intervals. When no coverage
success is observed, the estimate is not represented as an exact zero.
Instead, we report the one-sided $95\%$ upper confidence bound
\[
P_{\mathrm{cov}}^{\mathrm{up}}
=
1-0.05^{1/T},
\]
where $T$ is the final number of independent drops used at that operating
point. In particular, for a low-coverage point evaluated with $T=30000$
drops and no observed success, this gives
$P_{\mathrm{cov}}^{\mathrm{up}}\simeq 9.99\times10^{-5}$.
For the sample-mean quantities
$C_{\mathrm{erg}}$, $\mathbb{E}[I]$,
$\mathbb{E}[D_0]$, and $J_{\mathrm{SINR}}$, we report two-sided
$95\%$ Student-$t$ confidence intervals computed from the corresponding
final independent samples.

\begin{table}[!tb]
\centering
\caption{Baseline parameters for numerical illustration.}
\label{tab:sim_params}
\footnotesize
\setlength{\tabcolsep}{3pt}
\renewcommand{\arraystretch}{1.12}
\begin{tabular}{@{}p{0.34\columnwidth}p{0.59\columnwidth}@{}}
\toprule
Parameter & Value \\
\midrule

Earth radius
& $e=6371\,\mathrm{km}$ \\

Altitude/orbit radius
& $h=550\,\mathrm{km}$, $r=e+h$ \\

Inclination/user latitude
& $\phi=53^\circ$, $l_u=0^\circ$ \\

Path-loss exponent
& $\alpha=2.5$ \\

Fading power
& Rayleigh and Nakagami-$m$, $m=2$ \\

Antenna gains
& 
Baseline: $(g_r,g_t)=(1,10)$;
sensitivity: $(0.1,100)$ and $(0.01,1000)$;
$d_g=1000\,\mathrm{km}$ \\

SINR threshold
& $\tau=0\,\mathrm{dB}$, unless stated \\

Noise calibration
& $\mathrm{SNR}@d_{\min}=12\,\mathrm{dB}$,
$d_{\min}=r-e$ \\

Pilot MC drops
& 
$T_{\mathrm{pilot}}=6000$ per operating point, used only to identify
low-coverage points \\

Final MC drops
& 
$T=6000$ normally and $T=30000$ when fewer than $100$ coverage
successes occur in the independent pilot run \\

Confidence intervals
& 
Exact $95\%$ Clopper--Pearson intervals for coverage and
$95\%$ Student-$t$ intervals for sample means \\

\bottomrule
\end{tabular}
\end{table}

\paragraph*{Explicit annulus-block constants}

To identify the guaranteed finite-$N$ regime of
Lemma~\ref{lem:block}, we evaluate its geometric construction for
$h=550$~km, $\phi=53^\circ$, and $l_u=0^\circ$.
We choose $(\theta_0,\omega_0)=(\pi,\pi)$, for which
$d_0=r-e=550$~km, and set $\varepsilon_0=250$~km, yielding
$[d_1,d_2]=[300,800]$~km.
A deterministic search gives the nonwrapping rectangle
\[
\begin{aligned}
B
&=[\pi-0.047729,\pi+0.047729]\\&\times[\pi-0.047729,\pi+0.047729],
\end{aligned}
\]
with $\Delta_\theta=\Delta_\omega=0.095459$~rad.
Fine-grid evaluation together with a Lipschitz correction certifies that all
distances over $B$ lie in $[549.59,790.39]$~km and hence strictly inside the
selected annulus.

\begin{table}[!tb]
\centering
\caption{One certified annulus-block construction for the
baseline geometry.}
\label{tab:annulus_constants}
\footnotesize
\setlength{\tabcolsep}{3pt}
\renewcommand{\arraystretch}{1.10}
\begin{tabular}{@{}p{0.41\columnwidth}p{0.52\columnwidth}@{}}
\toprule
Quantity & Value \\
\midrule

Anchor phase
& $(\theta_0,\omega_0)=(\pi,\pi)$ \\

Anchor distance
& $d_0=550\,\mathrm{km}$ \\

Annulus margin
& $\varepsilon_0=250\,\mathrm{km}$ \\

Annulus
& $[d_1,d_2]=[300,800]\,\mathrm{km}$ \\

Rectangle side lengths
& $\Delta_\theta=\Delta_\omega=0.095459\,\mathrm{rad}$ \\

Certified distance range
& $[549.59,790.39]\,\mathrm{km}$ \\

Positive-density constant
& $\beta=5.770461\times10^{-5}$ \\

Grid threshold
& $n_0=132$ \\

Sufficient regime
& $\min\{\No,\Ns\}\ge132$ \\

\bottomrule
\end{tabular}
\end{table}

The rectangle gives
$\beta=\frac14
(\Delta_\theta/(2\pi))(\Delta_\omega/(2\pi))
=5.770461\times10^{-5}$ and
$n_0=\lceil\max\{4\pi/\Delta_\theta,4\pi/\Delta_\omega\}\rceil=132$.
Thus, the annulus-block property is guaranteed whenever
$\min\{\No,\Ns\}\ge132$.
For a balanced constellation, this implies
$N=\No\Ns\ge17\,424$.
This is a conservative sufficient threshold obtained from one admissible,
nonoptimized construction, rather than the empirical constellation size at
which coverage or spectral efficiency begins to decrease.

%-------------------------------------------------------------------------------
\subsection{Densification sweep and reuse policies}
\label{subsec:densif_sweep}

We sweep a standard densification family of Walker sizes
$\{(\No,\Ns)\}$ spanning
$N=\No\Ns$ from $\mathcal{O}(10^2)$ to $\mathcal{O}(10^4)$ satellites.
We compare four PRB activity policies:
\begin{enumerate}
\item \textbf{Full reuse:} $q=1$.
\item \textbf{Fixed thinning/reuse 10:} $q=0.1$.
\item \textbf{Fixed thinning/reuse 33:} $q=0.03$.
\item \textbf{Scaling policy:}
$q(N)=\min(1,c/N)$ with $c=600$, which enforces
$qN=\Theta(1)$ at large $N$.
\end{enumerate}

Each policy is evaluated under both Rayleigh fading and Nakagami-$m$ fading
with $m=2$, while all remaining geometry, propagation, noise, and
Monte-Carlo parameters are kept unchanged.

The coverage and rate plots also include guides proportional to $1/(qN)$ to
illustrate the order predicted by the converse bounds.

%-------------------------------------------------------------------------------
\subsection{Results: collapse under full reuse and the role of thinning}
\label{subsec:sim_results}

\paragraph*{Serving-distance shrinkage and initial gains}
As $N$ increases, the nearest visible satellite becomes closer, so the desired
signal initially strengthens.
Figure~\ref{fig:servdist} shows that $\mathbb{E}[D_0]$ decreases rapidly and
then approaches the geometric minimum.
This creates a transient regime in which coverage and
$C_{\mathrm{erg}}$ improve before interference becomes dominant.

\begin{figure}[!tb]
\centering
\includegraphics[width=0.96\linewidth]{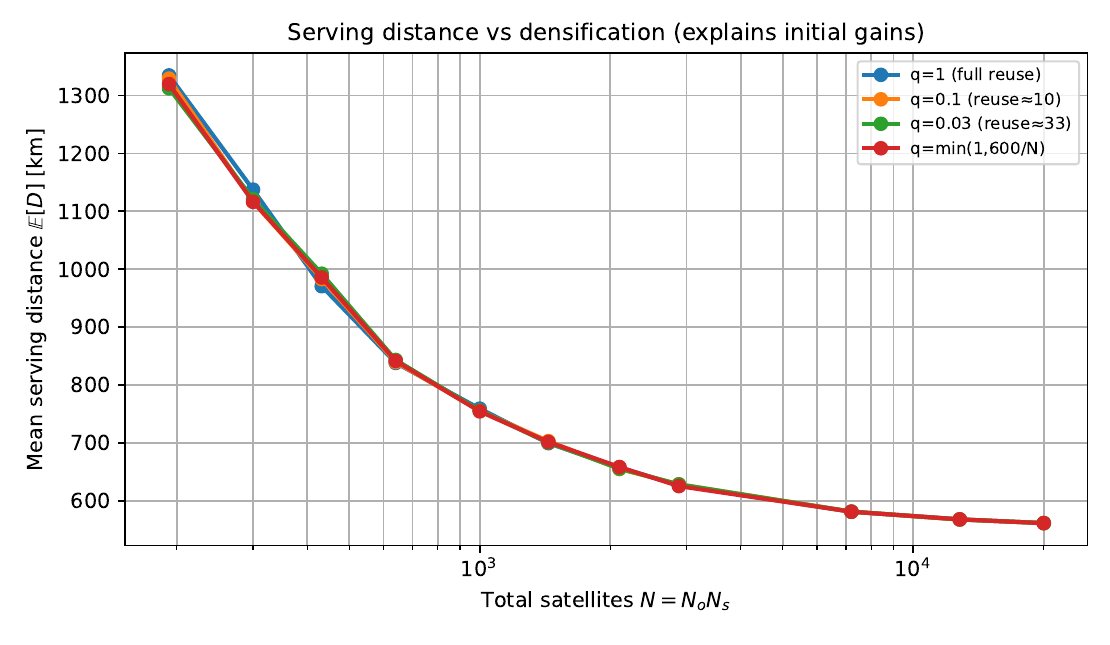}
\caption{Mean serving distance $\mathbb{E}[D_0]$ versus densification.
The distance initially decreases rapidly and then approaches the geometric
minimum, explaining the transient performance gain at moderate $N$.
Error bars are two-sided $95\%$ Student-$t$ confidence intervals computed
from the final independent Monte-Carlo samples.}
\label{fig:servdist}
\end{figure}

\paragraph*{Interference growth and the collapse mechanism}
Figure~\ref{fig:interf_growth} reports the mean interference
$\mathbb{E}[I]$.
Under full reuse, the interference grows approximately linearly with $N$,
consistently with the analytical lower bound.
Fixed thinning reduces the slope but does not eliminate the eventual growth.
In contrast, $q(N)=\min(1,c/N)$ causes $\mathbb{E}[I]$ to flatten at large
$N$, illustrating why keeping $qN=\Theta(1)$ prevents interference blow-up.

Because both fading distributions are normalized to unit mean, the scaling of
the mean interference is unchanged by the fading-law comparison.

\begin{figure}[!tb]
\centering
\includegraphics[width=0.96\linewidth]{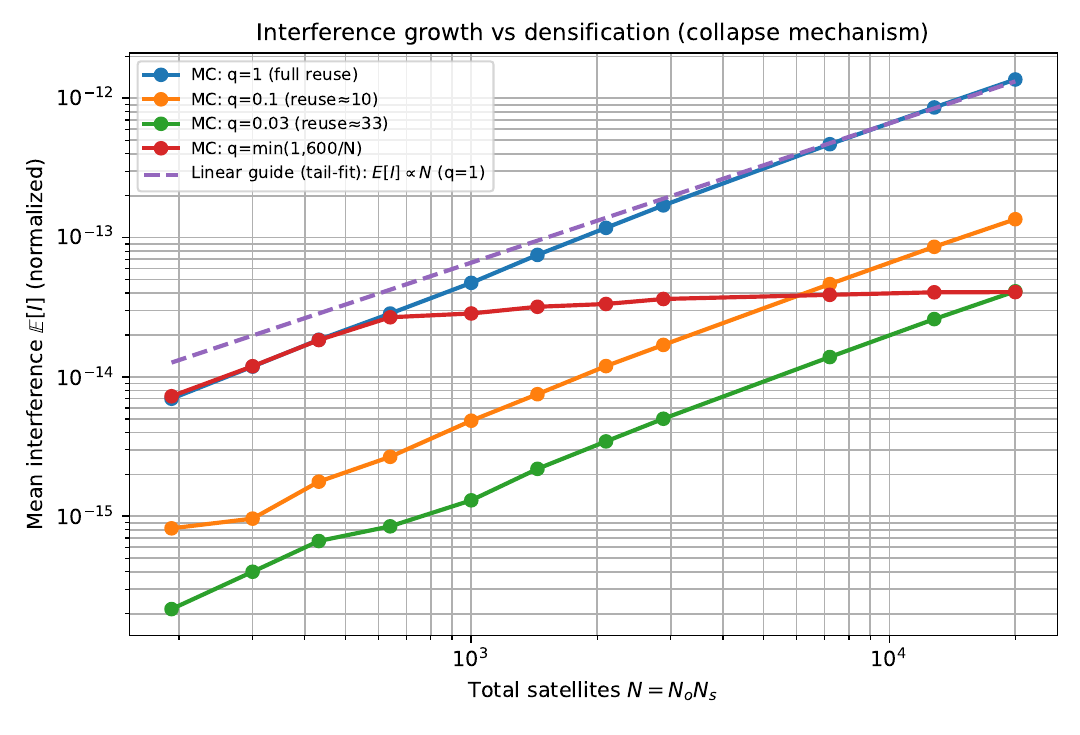}
\caption{Mean interference versus densification.
Full reuse exhibits approximately linear growth, fixed thinning reduces the
growth rate, and $q(N)=\min(1,c/N)$ produces a large-$N$ interference
plateau.
Error bars are two-sided $95\%$ Student-$t$ confidence intervals computed
from the final independent Monte-Carlo samples.}
\label{fig:interf_growth}
\end{figure}

\paragraph*{Coverage behavior}
Figure~\ref{fig:cov_densif} shows the coverage probability on a logarithmic
vertical axis.
Under full reuse, coverage initially increases and then collapses.
Fixed reuse shifts the densification knee toward larger constellation sizes,
while the scaling policy maintains nontrivial coverage throughout the
considered range.

Rayleigh and Nakagami-$m$ fading produce different finite-$N$ coverage values
and slightly different knees, but both preserve the same reuse-dependent
collapse and stabilization behavior.

\begin{figure}[!tb]
\centering
\includegraphics[width=0.96\linewidth]{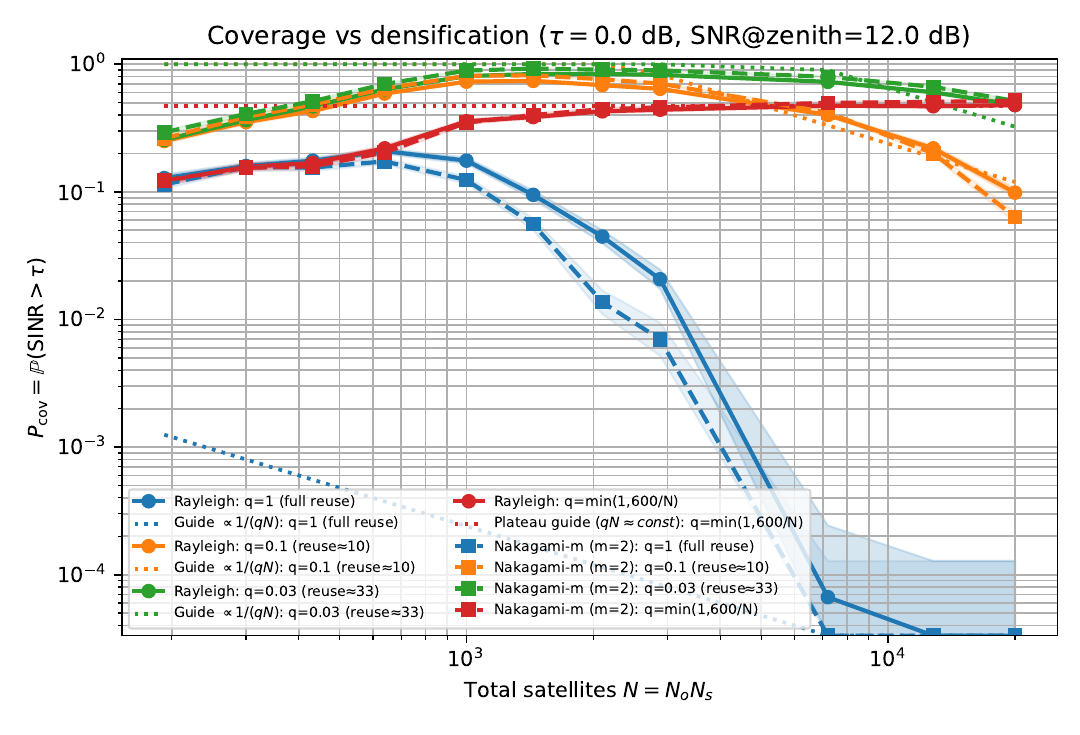}
\caption{Coverage probability versus densification.
Under full reuse, coverage reaches a finite-$N$ peak and then collapses.
Fixed reuse shifts the peak toward larger constellation sizes, whereas
$q(N)=\min(1,c/N)$ stabilizes the large-$N$ behavior.
Rayleigh and Nakagami-$m$ fading preserve the same qualitative trends.
Error bars are exact two-sided $95\%$ Clopper--Pearson confidence intervals.
When no coverage success is observed, an open marker reports the one-sided
$95\%$ upper confidence bound $1-0.05^{1/T}$ rather than an exact zero.}
\label{fig:cov_densif}
\end{figure}

\paragraph*{Ergodic spectral efficiency}
Figure~\ref{fig:cerg_densif} shows the corresponding behavior of the ergodic
spectral efficiency.
The serving-distance reduction initially improves performance, but the
increasing interference eventually dominates under full reuse.
Fixed reuse delays the degradation, while the scaling policy stabilizes the
large-$N$ rate.

The same qualitative conclusion holds under Rayleigh and Nakagami-$m$ fading:
the fading distribution changes the finite-$N$ values but not the
reuse-dependent scaling behavior.

\begin{figure}[!tb]
\centering
\includegraphics[width=0.96\linewidth]{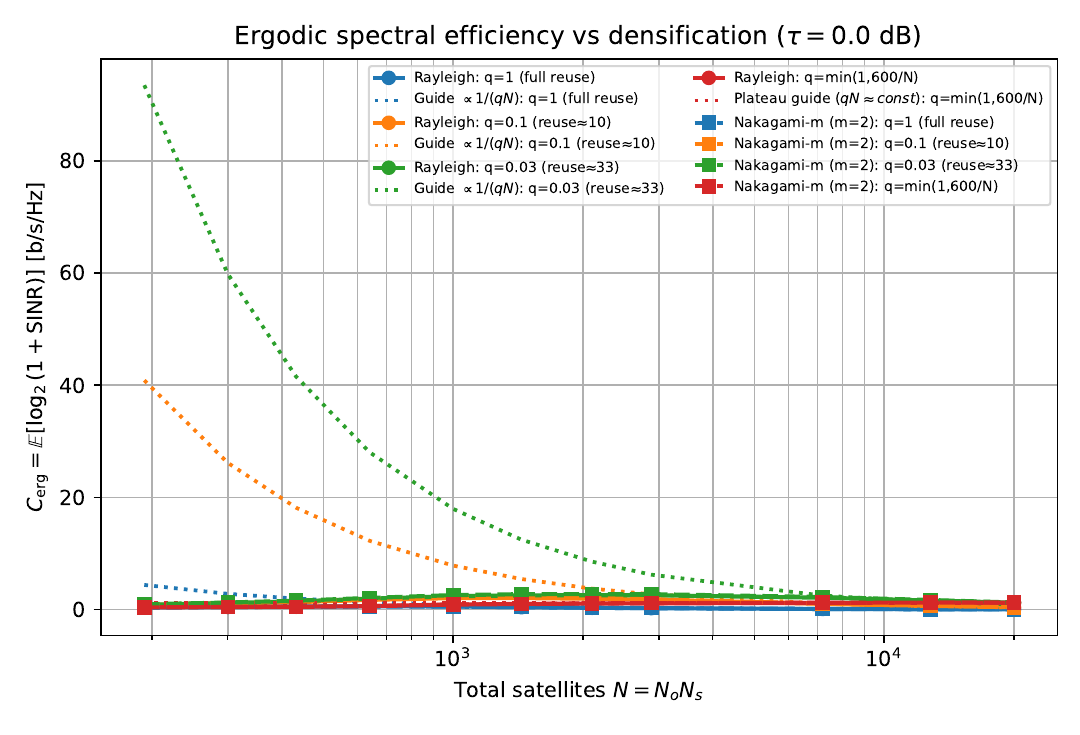}
\caption{Ergodic spectral efficiency versus densification.
The initial gain from the shrinking serving distance is eventually
overwhelmed under full reuse, whereas suitable reuse scaling stabilizes the
large-$N$ performance.
Rayleigh and Nakagami-$m$ fading exhibit the same qualitative behavior.
Error bars are two-sided $95\%$ Student-$t$ confidence intervals computed
from the final independent Monte-Carlo samples.}
\label{fig:cerg_densif}
\end{figure}

\paragraph*{Validation of the $1/(qN)$ order}
Figure~\ref{fig:scaling_check} reports the two normalized quantities
$(qN P_{\mathrm{cov}})/K_{\mathrm{cov}}$ and
$(qN C_{\mathrm{erg}})/K_{\mathrm{erg}}$ using the conservative theorem
constants.
The ratios remain below one, often by several orders of magnitude, confirming
that the bounds capture the correct order while remaining conservative.
This slack is expected because the converse is based on a worst-case
annulus-block construction rather than designed as a tight performance
predictor.

\begin{figure*}[!t]
\centering
\includegraphics[width=0.86\textwidth]{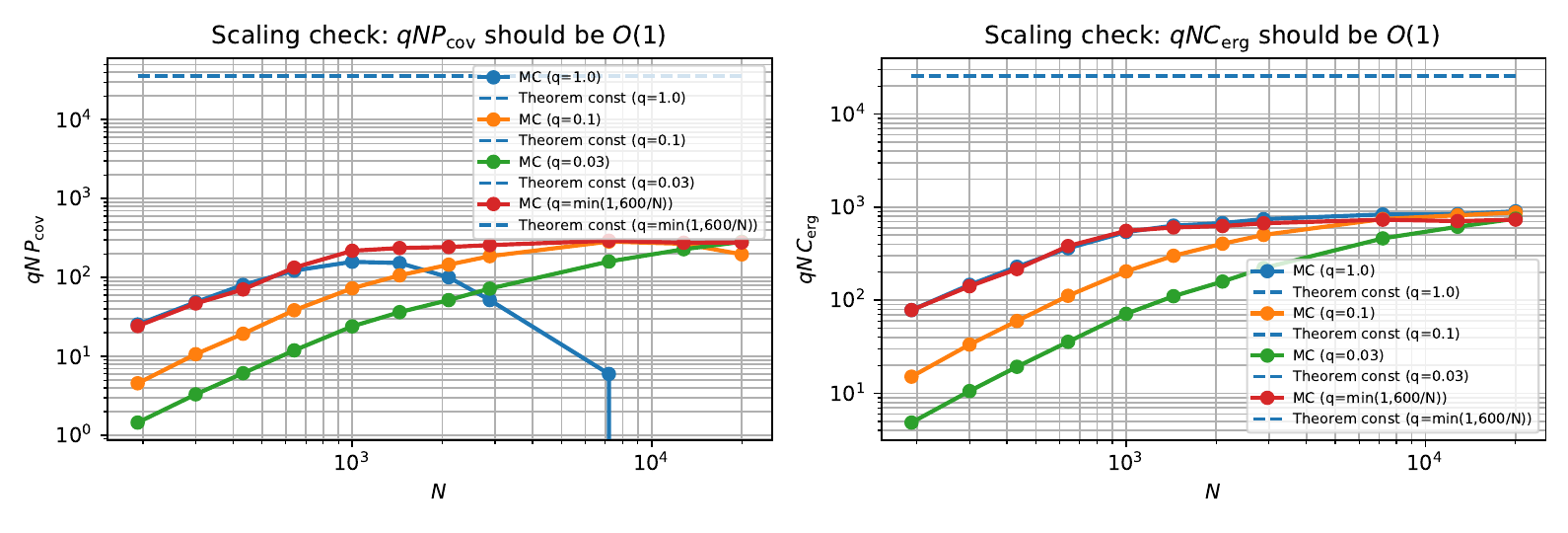}
\caption{Scaling validation against the conservative theorem constants.
The left panel reports $(qN P_{\mathrm{cov}})/K_{\mathrm{cov}}$, and the
right panel reports $(qN C_{\mathrm{erg}})/K_{\mathrm{erg}}$.
The slack reflects the converse nature of the bounds while confirming the
$1/(qN)$ order.}
\label{fig:scaling_check}
\end{figure*}

%-------------------------------------------------------------------------------
\subsection{Mean-SINR optimum and finite-search validation}
\label{subsec:mean_sinr_tightness}

Theorem~\ref{thm:exist_opt_sinr} complements the coverage- and rate-collapse
results by establishing that the invariant-measure averaged mean SINR
eventually decreases as $O(1/(\No\Ns))$. Consequently, its maximization over
the admissible Walker parameters reduces to a finite search. We now illustrate
this design implication numerically under the baseline full-reuse Rayleigh
setting.

For this experiment, conditional Monte Carlo is used to obtain an accurate
estimate of
$J_{\mathrm{SINR}}(\No,\Ns;l_u)=\mathbb{E}[\mathrm{SINR}]$.
The serving-link Rayleigh fading is averaged analytically using
$\mathbb{E}[H_0]=1$, while antithetic exponential samples are used for the
interfering-link fadings. Exactly $T=6000$ independent conditional-Monte-Carlo
realizations are used at every balanced-design operating point. The balanced
Walker-design grid is refined around the observed optimum, and two-sided
$95\%$ Student-$t$ confidence intervals are reported at every operating
point.

As shown in the left panel of
Fig.~\ref{fig:mean_sinr_tightness}, the mean SINR initially increases with
densification because the serving distance decreases. It reaches its maximum
over the simulated balanced-design grid at
$(\widehat N_o^\star,\widehat N_s^\star)=(19,19)$, corresponding to
$\widehat N^\star=361$, with
$\widehat J_{\mathrm{SINR}}^\star\simeq2.572$. Beyond this operating point,
the growth of the co-channel interference dominates the remaining serving-link
gain, and the mean SINR decreases steadily.

\begin{figure*}[!t]
\centering
\includegraphics[width=0.94\textwidth]{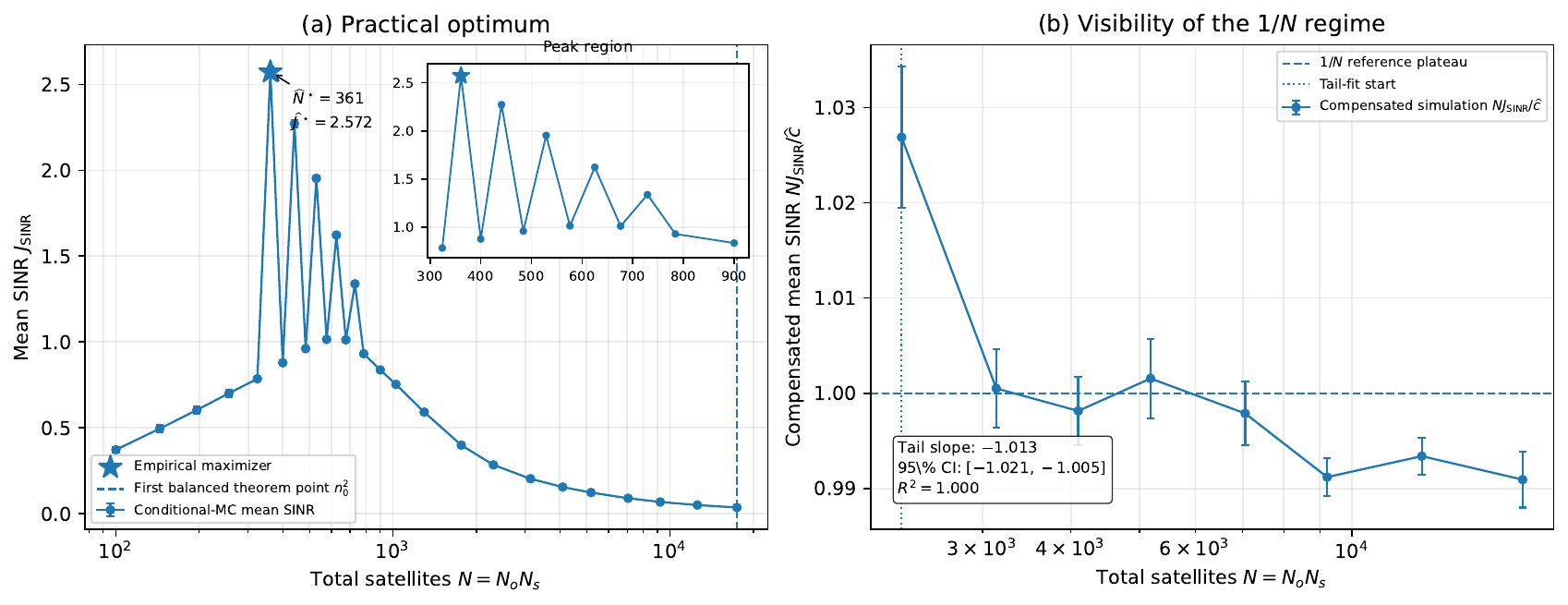}
\caption{
Finite-density mean-SINR optimum and validation of the predicted decay order
under the baseline full-reuse Rayleigh setting.
The left panel reports the conditional-Monte-Carlo estimate of
$J_{\mathrm{SINR}}$ over the balanced Walker-design grid, together with
two-sided $95\%$ Student-$t$ confidence intervals and the empirical maximizer
$(\widehat N_o^\star,\widehat N_s^\star)=(19,19)$.
The dashed vertical line marks the first balanced constellation
$N=n_0^2=17\,424$ satisfying the explicit phase-state--uniform geometric
condition.
The right panel reports the compensated quantity
$N J_{\mathrm{SINR}}/\widehat c$, where $\widehat c$ is obtained from the
post-peak samples. Its nearly constant behavior confirms the
$1/N$ mean-SINR scaling.}
\label{fig:mean_sinr_tightness}
\end{figure*}

The dashed line in the left panel marks
$N=n_0^2=132^2=17\,424$, the first balanced constellation satisfying the
explicit condition $\min\{\No,\Ns\}\ge n_0$ of
Lemma~\ref{lem:block}. This threshold is a sufficient certification point
arising from a phase-state--uniform geometric construction. The theorem and
the numerical search therefore play complementary roles: the theorem provides
a rigorous finite-search guarantee that is uniform over the phase state,
whereas the numerical evaluation identifies the relevant operating point for
the selected baseline parameters.

The post-peak samples provide a direct validation of the analytical scaling
law. Fitting $J_{\mathrm{SINR}}(N)\approx aN^s$ over the simulated tail gives
\[
\widehat s=-1.013,
\qquad
95\%~\mathrm{CI}=[-1.021,-1.005],
\qquad
R^2\simeq1.
\]
Accordingly, the compensated quantity
$N J_{\mathrm{SINR}}/\widehat c$ in the right panel remains close to one.
Thus, the $O(1/N)$ decay established by
Theorem~\ref{thm:exist_opt_sinr} is already clearly visible over the simulated
constellation range.

These results strengthen the finite-density design interpretation of the
converse: densification produces a well-defined pre-collapse operating region,
while the mean SINR subsequently follows the decay order predicted by the
theory. Hence, Theorem~\ref{thm:exist_opt_sinr} supplies the rigorous
finite-search foundation, and the associated numerical search identifies the
baseline-dependent maximizing balanced design.

%-------------------------------------------------------------------------------
\subsection{Sensitivity to sidelobe suppression}
\label{subsec:sidelobe_sensitivity}

We finally examine whether the finite-$N$ behavior is sensitive to the
sidelobe level of the two-tier antenna model.
Since $G_{\mathrm{main}}=g_tg_r$ and $G_{\mathrm{side}}=g_r$, changing $g_r$
while holding $g_t$ fixed would also alter the main-tier gain.
We therefore keep $G_{\mathrm{main}}=10$ and the zenith SNR unchanged and
compare $g_r\in\{1,0.1,0.01\}$, with
$g_t=G_{\mathrm{main}}/g_r$.
These settings correspond  to relative sidelobe levels of
$-10$, $-20$, and $-30$~dB, respectively.
The comparison uses full reuse, Rayleigh fading, nearest-visible association,
and the remaining parameters of Table~\ref{tab:sim_params}.

\paragraph*{Coverage sensitivity}

Figure~\ref{fig:sidelobe_coverage} shows that reducing the relative sidelobe
level from $-10$ to $-20$~dB increases the finite-$N$ coverage probability.
The additional improvement at $-30$~dB is smaller.
The empirical peaks remain in approximately the same densification region,
and all cases eventually decline under full reuse.

\begin{figure}[!tb]
\centering
\includegraphics[width=0.96\linewidth]
{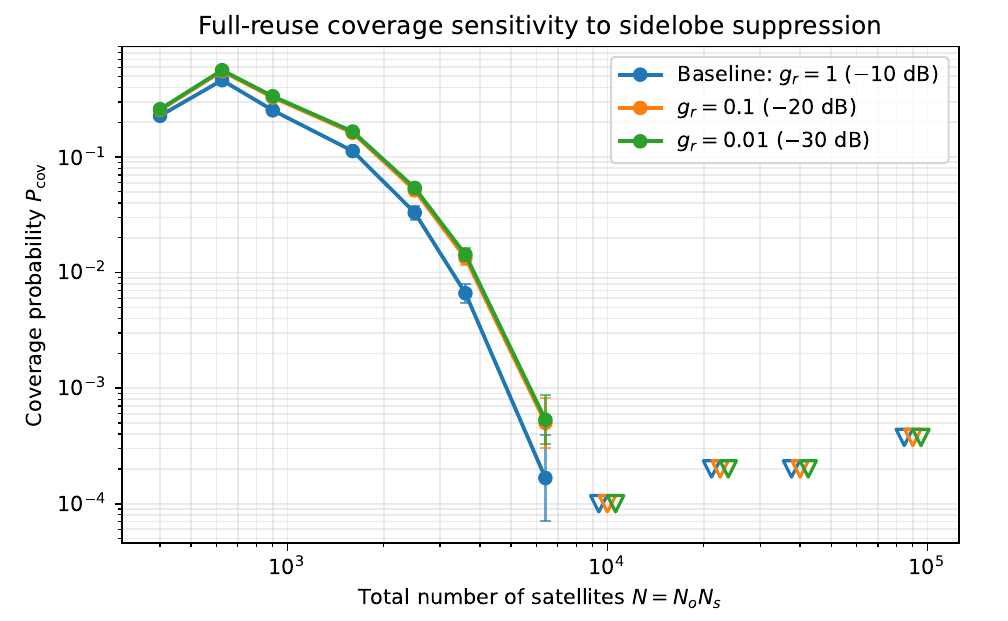}
\caption{
Full-reuse coverage sensitivity to sidelobe suppression under Rayleigh fading.
The main-tier gain and zenith SNR are fixed.
Reducing the relative sidelobe level from $-10$ to $-20$~dB improves
finite-$N$ coverage, whereas the additional improvement at $-30$~dB is
limited.
Error bars are exact two-sided $95\%$ Clopper--Pearson confidence intervals.
When no coverage success is observed, open triangles denote the one-sided
$95\%$ upper confidence bound $1-0.05^{1/T}$ rather than an exact zero;
coincident markers are slightly shifted horizontally for visibility.}
\label{fig:sidelobe_coverage}
\end{figure}

\paragraph*{Rate sensitivity}

Figure~\ref{fig:sidelobe_rate} shows the corresponding ergodic
spectral-efficiency behavior.
Reducing the relative sidelobe level from $-10$ to $-20$~dB improves the
finite-$N$ rate, whereas the additional improvement at $-30$~dB is limited.
All three curves retain the eventual interference-driven degradation.

\begin{figure}[!tb]
\centering
\includegraphics[width=0.96\linewidth]
{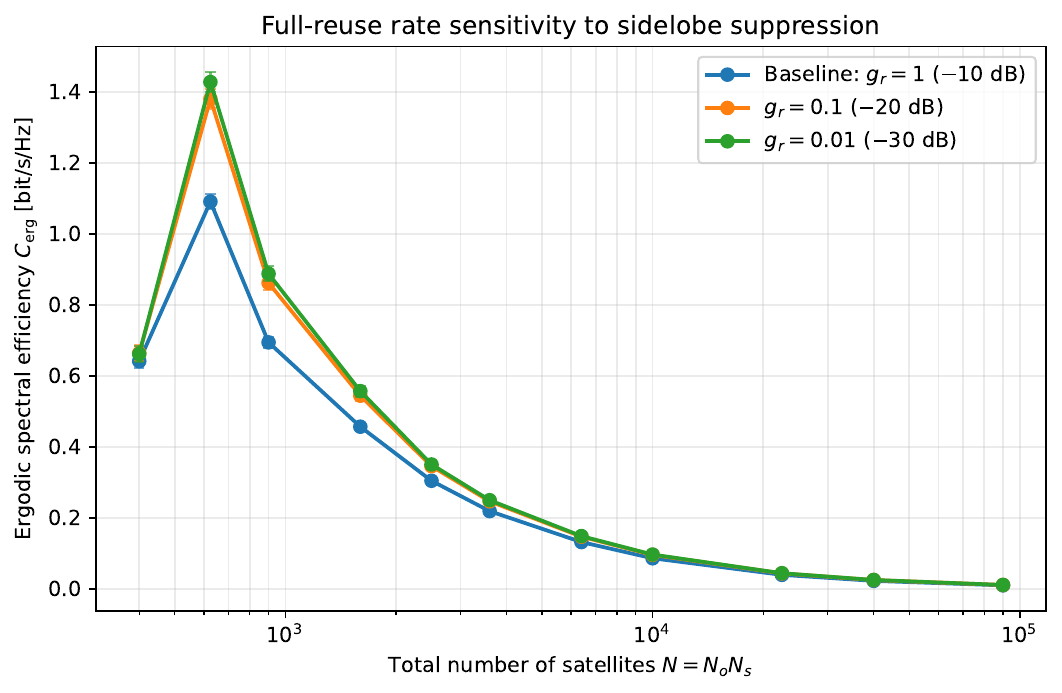}
\caption{
Full-reuse ergodic spectral-efficiency sensitivity to sidelobe suppression.
Reducing the relative sidelobe level from $-10$ to $-20$~dB improves the
finite-$N$ rate, whereas the additional improvement at $-30$~dB is limited.
All cases eventually decline under full reuse.
Error bars are two-sided $95\%$ Student-$t$ confidence intervals computed
from the final independent Monte-Carlo samples.}
\label{fig:sidelobe_rate}
\end{figure}

\paragraph*{Interference decomposition}

The diminishing improvement beyond $-20$~dB is explained by
Fig.~\ref{fig:sidelobe_fraction}.
Under the baseline setting, links with $d>d_g$ contribute a nonnegligible
fraction of the aggregate interference.
This fraction decreases to only a few percent at $-20$~dB and becomes nearly
negligible at $-30$~dB.
The remaining interference is therefore dominated by the growing co-channel
population with $d\le d_g$, whose gain is fixed in this experiment.
Consequently, sidelobe suppression improves the finite-$N$ constants but does
not change the $O(1/N)$ and $O(1/(qN))$ collapse orders.

\begin{figure}[!tb]
\centering
\includegraphics[width=0.96\linewidth]
{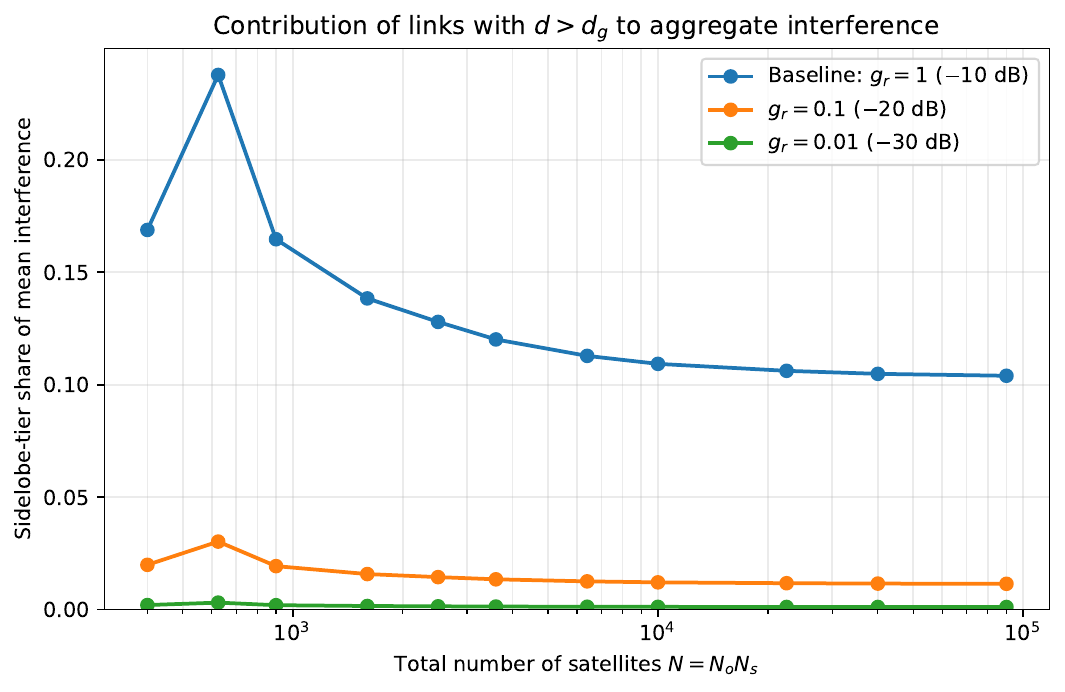}
\caption{
Fraction of the mean aggregate interference contributed by links with
$d>d_g$.
At $-20$~dB, the sidelobe tier already contributes only a few percent,
explaining the limited additional improvement at $-30$~dB.
Error bars are two-sided $95\%$ Student-$t$ confidence intervals computed
from the final independent Monte-Carlo samples.}
\label{fig:sidelobe_fraction}
\end{figure}

\noindent\textbf{Reproducibility.}

All figures are generated by the accompanying Python scripts using
Monte-Carlo sampling over the phase shifts, fading variables, and
per-resource-block activity, conditional Monte Carlo for the mean-SINR study,
and the deterministic annulus-constant calculation described above.
For every operating point, the scripts record the pilot classification, the
final sample size $T$, the number of observed coverage successes, the point
estimate, and the corresponding confidence-interval endpoints. The pilot
samples used for low-coverage classification are generated independently and
are excluded from all reported estimates.

\section{Conclusion}

This paper provided a set of \emph{densification converses} for structured
Walker LEO constellations when the downlink operates under nearest-visible
association and, unless stated otherwise, full frequency reuse on a common
time--frequency resource block. Our performance interpretation is explicitly
tied to the \emph{invariant (stationary) measure} on the Walker phase cell,
which is the natural distribution induced by the constellation dynamics in the
rotating Earth frame. Within this physically meaningful stationary regime, we
identified a simple but fundamental message: \emph{densification alone is not a
remedy under full reuse}. Increasing the total satellite count
$N=\No\Ns$ increases the number of simultaneously visible co-channel
transmitters, and the resulting interference growth eventually dominates the
beneficial reduction in serving distance.

The central technical mechanism is the phase-state--uniform annulus-block
property of Lemma~\ref{lem:block}. It guarantees that, for every phase state and
all sufficiently large constellations, a positive-density fraction of visible
satellites remains inside a fixed distance annulus strictly within the horizon.
This property converts the geometric interference mechanism into explicit
finite-$N$ bounds that are uniform over the invariant-measure state space. The
resulting collapse concerns link-level SINR, coverage, and spectral efficiency
on the considered shared resource block; it does not imply that the aggregate
capacity of the complete multi-resource network vanishes.

The principal design implication concerns reuse planning. Under independent
per-resource-block activity with probability $q$, maintaining nonvanishing
coverage requires $qN=O(1)$, or equivalently
$F=1/q=\Omega(N)$. Thus, a fixed reuse factor can delay the
interference-dominated regime but cannot prevent its eventual onset as the
constellation grows. Sustaining nontrivial per-resource-block performance
instead requires reuse, scheduling, beam hopping, coordinated muting, or
another interference-management mechanism that keeps the effective number of
co-channel visible satellites bounded. The explicit $O(1/(qN))$ bound provides
a direct necessary guideline for scaling this activity level with the
constellation size.

The same geometric mechanism also shows that the mean SINR eventually decreases
with densification, so its maximization over the admissible Walker parameters
reduces to a finite search. Operationally, this suggests selecting a
constellation size within the finite pre-collapse regime and increasing reuse
or scheduling protection as additional satellites are deployed.

The numerical illustration in Section~\ref{sec:numerical} supports this picture
at finite $N$: serving distances initially shrink with densification, but under
full reuse the interference grows nearly linearly in $N$, producing a
coverage/rate knee and eventual collapse. Enforcing the scaling
$qN=\Theta(1)$ instead produces an interference plateau and stabilizes
performance.

Future work should complement these necessary scaling laws with achievable
deterministic reuse and scheduling constructions for the Walker grid. In
particular, an important direction is to develop dynamics-aware frequency-reuse
policies in which the frequency or resource-block assignment adapts
deterministically to the evolving Walker phase state, rather than relying on
independent random thinning. Such policies could exploit the predictable orbital
dynamics to separate satellites that become strong co-channel interferers over
time and may lead to substantially more efficient reuse patterns. Other relevant
extensions include realistic multi-beam antenna patterns, explicit gateway and
traffic constraints, power and bandwidth adaptation, and more general
heterogeneous constellations with a growing number of shells or altitudes that
scale with the constellation size.

\end{document}